\documentclass[11pt,english]{article}

\usepackage{times}
\usepackage{amsfonts}
\usepackage{amsmath,amssymb}
\usepackage{amsthm}
\usepackage{natbib}
\usepackage{bm}
\usepackage{graphicx}
\usepackage{subcaption}
\usepackage{appendix}
\usepackage{rotating}
\usepackage{float}
\usepackage{mathrsfs}
\usepackage{dsfont}
\usepackage{nicefrac}
\usepackage{bm}
\usepackage{xspace}
\usepackage{color}
\usepackage[colorlinks,linkcolor=blue,citecolor=blue,bookmarks=false,pagebackref]{hyperref}
\usepackage[top=1.0in, bottom=1.3in, left=1.15in, right=1.15in]{geometry}
\usepackage{setspace}
\usepackage{titling}
\usepackage[final]{pdfpages}
\allowdisplaybreaks

\theoremstyle{plain}
\newtheorem{theorem}{Theorem}[section]
\theoremstyle{plain}
\newtheorem{lemma}{Lemma}[section]
\theoremstyle{plain}
\newtheorem{assumption}{Assumption}[section]
\theoremstyle{plain}

\theoremstyle{plain}
\newtheorem{corollary}{Corollary}[section]
\theoremstyle{remark}
\newtheorem{remark}{Remark}[section]
\theoremstyle{definition}
\newtheorem{definition}{Definition}[section]

\newcommand{\argmin}{\operatornamewithlimits{argmin\,}}

\newcommand{\E}{\mathbb{E}}

\newcommand{\has}{\mathcal{H}}
\newcommand{\was}{\mathcal{W}}
\newcommand{\pdist}{\mathcal{D}}
\mathchardef\mhyphen="2D
\def \d {\mathrm{d}}

\floatstyle{ruled}
\newfloat{algorithm}{tbp}{loa}
\providecommand{\algorithmname}{Algorithm}
\floatname{algorithm}{\protect\algorithmname}

\title{Robust and Efficient Approximate Bayesian Computation: A Minimum Distance Approach}
\author{David T. Frazier\thanks{Australian Research Council and Department of Econometrics and Business Statistics, Monash University, Australia.
Address correspondence to david.frazier@monash.edu.}}
\date{}

\begin{document}

\maketitle

\begin{abstract}
\noindent  In many instances, the application of approximate Bayesian methods is hampered by two practical features: 1) the requirement to project the data down to low-dimensional summary, including the choice of this projection, and which ultimately yields inefficient inference; 2) a possible lack of robustness of these methods to deviations from the underlying model structure. Motivated by these efficiency and robustness concerns, we construct a new Bayesian method that can deliver efficient estimators when the underlying model is well-specified, and which is simultaneously robust to certain forms of model misspecification. This new approach bypasses the calculation of summaries by considering a norm between empirical and simulated probability measures. For specific choices of the norm, we demonstrate that this approach can deliver point estimators that are as efficient as those obtained using exact Bayesian inference, while also simultaneously displaying robustness to deviations from the underlying model assumptions.
\end{abstract}


\section{Introduction \label{sec:introduction}}

The use of complex models has led researchers to construct statistical methods that are capable of conducting reliable inference in situations where the likelihood function of such models is intractable. In Bayesian inference, the class of approximate Bayesian computation (ABC) methods,  which replaces evaluation of the likelihood with model simulation, are often employed; for a review on ABC methods see the recent handbook \cite{sisson2018handbook}.

ABC is predicated on the belief that the observed data
is drawn from a given class of models from which it is feasible to generate synthetic data. This fact is then used to construct an approximate posterior for the model unknowns by comparing, in some given distance measure, quantities calculated using observed data and data simulated from the model. While such a procedure allows us to conduct inference in very complex models, the broad applicability of ABC methods has been hindered by at least two factors. First, the choice of summary measures, and as a by product the resulting loss of information due to such a choice. Two, the reliability and robustness of such procedures to deviations from the modeling assumptions.

In many cases, the quantities used to compare observed and synthetic data correspond to moments or sample quantiles. Currently, there is no singularly accepted choice for such ``summary statistics'', as different problems require different approaches, and researchers are therefore left to fish for useful summaries. Regardless of the choice of summaries, with few exceptions, the use of such summary measures will inevitably result in a loss of information relative to a procedure that uses the entire sample, such as (infeasible) exact Bayes approaches. More formally, in a frequentest sense, point estimators obtained from ABC procedures will generally be inefficient (see, e.g., \citealp{LF2016} and \citealp{FMRR2016}).

In addition, it is known since \cite{frazier2020model} that different versions of ABC can deliver very different point estimators under even mild levels of model misspecification. Moreover, the results of these authors suggest that the more complicated the ABC approach, such as, e.g., nonlinear regression adjustments ABC (\citealp{blumF2010non}), the poorer the performance. In this way, the authors argue Occam's razor should be applied to ABC inference when the idea of model misspecification is entertained.  

Motivated by these practical concerns of efficiency and robustness of ABC inference, in Section \ref{sec:mde} we propose a novel ABC approach that can deliver estimators that are as efficient as the infeasible exact Bayes approach, when the model generating synthetic data is correctly specified, but simultaneously displays robustness  to deviations from the modeling assumptions. To avoid the use of summarization via summary statistics, we follow the idea of \cite{Bernton2017} and view the observed and synthetic data through their empirical probability measures, and propose the use of a norm based on these measures to measure the discrepancy between the observed and synthetic datasets. It is well known that such norms give rise to so-called minimum distance functionals, \cite{donoho1988automatic}, and so we term this new approach to ABC, minimum distance ABC (MD-ABC). The MD-ABC posterior is shown to be asymptotically normal under regularity conditions that are known to be satisfied for specific choices of the norm. Moreover, if one takes as the norm the Hellinger distance, the resulting ABC approach delivers point estimators that are asymptotically efficient. The later result extends to the Bayesian framework the results on minimum Hellinger distance estimation (MHDE) pioneered by \cite{beran1977minimum}.

Furthermore, in Section \ref{sec:robust} we demonstrate that, in the class of ABC point estimators, MD-ABC is optimally robust to certain classes of deviations from the assumed model. Using a local model misspecification framework (see, e.g., \citealp{rieder2012robust}), we demonstrate that MD-ABC can display the smallest possible bias and mean squared error in the class of ABC point estimators. This result therefore extends the known robustness of frequentest based MD estimators, see, e.g., \cite{donoho1988automatic}, in two novel directions: 1) to situations where inference is carried out in the Bayesian framework; 2) to situations where MD estimation must be implemented via simulation (i.e., where the model is intractable). Simulated examples considered in Section \ref{sec:robust_ex} demonstrate that in locally misspecified models, MD-ABC can produce superior inference to exact Bayes inference and competing ABC approaches, such as ABC based on the Wasserstein distance.

\section{Implicit models, and ABC Inference\label{sec:generativemodels}}
\subsection{Implicit Models}
Let $(\Omega,\mathcal{F},\mathbb{P})$ denote the intrinsic probability space, with associated expectation operator $\mathbb{E}$, on which all the random variables are defined. Denote by $\mathcal{P}(\mathcal{X})$ the set of probability measures on a space $\mathcal{X}$. We observe $y_{1:n} = (y_1,\ldots,y_n)'$ data points, taking values in $\mathcal{Y}^n$, and which is distributed according to some unknown measure $\mu_\star^{(n)}$. 

In most cases, the space $\mathcal{P}(\mathcal{X})$ is specified to be a parametric class of models $$\mathcal{M}^{(n)}=\{\mu^{(n)}_\theta : \theta \in \Theta\} \subset
\mathcal{P}(\mathcal{Y}^n),\text{ where }\Theta\subset \mathbb{R}^{d_\theta},$$ and where  $\Pi\in\mathcal{P}(\Theta)$ quantifies our prior beliefs about the unknown parameters $\theta$. Regardless of the model complexity, it is often feasible to generate
synthetic observations $z_{1:n}$ according to $\mu^{(n)}_\theta$, for any $\theta\in\Theta$, but there are many cases where it is not feasible to evaluate the likelihood associated to $\mu^{(n)}_\theta$. Even in cases where the likelihood is infeasible to calculate, useful information about the model can still be obtained by comparing the observed data, $y_{1:n}$, and the simulated data, $z_{1:n}$. 

For $A\in\mathcal{F}$, the empirical measure of $A$ is $\hat{\mu}_n(A) = n^{-1}\sum_{i=1}^n \delta_{y_i}(A)$, where $\delta_y$ denotes the Dirac measure on $y\in\mathcal{Y}$. Likewise, for any $\theta\in\Theta$, let $\hat{\mu}_{\theta,n}(A)=n^{-1}\sum_{i=1}^n \delta_{z_i}(A)$ with $z_{1:n}\sim \mu_\theta^{(n)}$, denote the simulated empirical measure. Under weak conditions on the observed data and model class $\mathcal{M}^{(n)}$, the sequence $\{\hat{\mu}_{\theta,n}\}_{n\geq 1}$ of random probability measures, on $\mathcal{Y}^n$, will converge (weakly) toward $\mu^{(n)}_\theta\in\mathcal{P}(\mathcal{Y}^n)$, and, likewise, the sequence $\{\hat{\mu}_n\}_{n\geq1}$ will converge toward $\mu^{(n)}_\star\in  \mathcal{P}(\mathcal{Y}^n)$. Even if the likelihood is intractable, $\hat\mu_n$ and $\hat\mu_{\theta,n}$ can be easily constructed.

\subsection{Standard Approximate Bayesian Computation \label{sec:ABC}}

The following algorithm describes the general approach to approximate Bayesian computation. Let $\epsilon>0$ 
denote some pre-defined threshold, and let $\pdist$ denotes
an arbitrary distance measure between two data sets $y_{1:n}$ and $z_{1:n}$. The most 
\begin{enumerate}
	\item Draw a parameter $\theta$ from the prior distribution $\Pi$, and generate a synthetic dataset $z_{1:n} \sim \mu_{\theta}^{(n)}$.
	\item If $\pdist(y_{1:n}, z_{1:n})\leq\epsilon$, keep $\theta$, otherwise reject it.
\end{enumerate}
Accepted draws from the above can be shown to from the ABC posterior distribution: for $A\subset\Theta$, 
\begin{equation}\label{eq:abcposterior}
\Pi_\epsilon(\theta\in A|{y_{1:n}})= \int_{A}\frac{\d\Pi(\theta)\int_{\mathcal{Y}^n}\mathds{1}\left[\pdist(y_{1:n},z_{1:n})\leq\epsilon\right]\d\mu_\theta^{(n)}( z_{1:n})}{\int_{{\Theta}}\d\Pi(\theta)\int_{\mathcal{Y}^n}\mathds{1}\left[\pdist(y_{1:n},z_{1:n})\leq\epsilon\right]\d\mu_\theta^{(n)}(z_{1:n})},
\end{equation}
where $\mathds{1}$ is the indicator function.

In most cases, using a distance calculated on the data directly is infeasible. The most common approach is to first degraded the data down to a low-dimensional vector of summary statistics and replace $\mathcal{D}(y_{1:n},z_{1:n})$ with a distance based on the summaries. For $\eta:\mathcal{Y}^n\rightarrow\mathcal{B}\subset\mathbb{R}^{d_\eta}$, with $d_\eta$ much smaller than $n$, ABC is most commonly implemented using $\pdist(y_{1:n},z_{1:n}) = \|\eta(y_{1:n})-\eta(z_{1:n})\|$, and where $\|\cdot\|$ denotes the Euclidean norm.  

While practically useful, replacing the data with summaries creates several issues. Firstly, by the nature of the models to which ABC is generally applied, no set of sufficient statistics exist and  the use of summaries inevitably leads to a loss of information. Consequently, as $\epsilon\rightarrow0$ the ABC posterior converges to the partial posterior, i.e., the posterior of $\theta$ given $\eta(y_{1:n})$,
and, as shown in \cite{LF2016} and \cite{FMRR2016}, point estimators for $\theta$ obtained using ABC based on summary statistics will always have a larger variance than those obtained by exact Bayes. Thirdly, it is well-known that in certain situations the use of different summaries in ABC can lead to significantly different results. 

To circumvent the loss of information incurred by using summary statistics in ABC, \cite{Bernton2017} propose Wasserstein-based ABC (WABC): replace $\mathcal{D}(y_{1:n},z_{1:n})$ with the $p$-Wasserstein distance
\begin{equation} \label{eq:wass_def_assignment}
\was_p(\hat{\mu}_n, \hat\mu_{\theta,n})^p = \inf_{\sigma \in \mathcal{S}_n} \frac{1}{n} \sum_{i=1}^n \gamma(y_i, z_{\sigma(i)})^p, 
\end{equation}
where $\mathcal{S}_n$ is the set of permutations of $\{1, \dots,n\}$. In the case where the data is scalar valued, and if $\gamma(y_i,z_j) = |y_i - z_j|$, the infimum is achieved by sorting 
$y_{1:n}$ and $z_{1:n}$ in increasing order and matching the order statistics. The use of order statistics within ABC is well-founded,  see
e.g. \citet{FP2012}, and in the multivariate case Wassertein ABC can be interpreted as generalizing the use of order statistics within ABC to arbitrary dimension in a principled manner.

While undoubtedly novel, practical application of WABC can be hindered by the specialized tools required to manipulate the Wasserstein distance in cases other than $p=1$ and in the multivariate case. In addition, the nature of the Wasserstein distance makes analyzing the theoretical behavior of WABC difficult. In particular, the asymptotic shape of the WABC posterior is currently unknown, as is the behavior, and in particular the asymptotic efficiency of, point estimators obtained by from WABC. Lastly, it is unclear precisely how the deviations from the model affect the resulting inference in WABC. 

Regardless of the above, the ability of WABC to use the entire data has been shown to lead to strong efficiency gains over summary based ABC. Herein, we propose an alternative to WABC that obviates the need to calculate summaries of the data. However, unlike WABC, or ABC based on summary statistics, this new approach can deliver efficient inference in the case of correct model specification, and is simultaneously robust to deviations from the underlying modeling assumptions.

\section{Minimum Distance Based ABC Inference}\label{sec:mde}
Replacing the dataset with a low-dimensional summary statistic inevitable leads to an information loss, and a significant reduction in the efficiency of the resulting inference. Moreover, as discussed in \cite{frazier2020model}, different choices of summary statistics can lead to substantially different behavior when the assumed model is not correct, i.e., inference based on summary statistics are not ``stable'' under model misspecification. With these two features in mind, efficiency and stability of the resulting inferences, we propose a new approach to Bayesian inference that does not resort to a low-dimensional summary statistics but is instead based on distances between observed and simulated empirical measures. Let $\|\cdot\|_\mathcal{P}:\mathcal{P}(\mathcal{Y})\times \mathcal{P}(\mathcal{Y})\rightarrow\mathbb{R}_{+}$ denote a pseudo-norm on $\mathcal{P}(\mathcal{Y})$. Replacing  $\pdist$ by $\|\hat{\mu}_{\theta,n}-\hat{\mu}_{n}\|_\mathcal{P}$ in equation \eqref{eq:abcposterior} yields the ABC posterior
\begin{equation} \label{eq:Habcposterior}
  \Pi_\epsilon(A|{y_{1:n}})= \frac{\int_{A}\d\Pi(\theta)\int_{\mathcal{Y}^n}\mathds{1}\left[\|\hat{\mu}_{\theta,n}-\hat{\mu}_{n}\|_\mathcal{P}\leq\epsilon\right]\d\mu_\theta^{(n)}( z_{1:n})}{\int_{{\Theta}}\d\Pi(\theta)\int_{\mathcal{Y}^n}\mathds{1}\left[\|\hat{\mu}_{\theta,n}-\hat{\mu}_{n}\|_\mathcal{P}\leq\epsilon\right]\d\mu_\theta^{(n)}(z_{1:n})}.
\end{equation}

 It is important to point out that	MD-ABC is not limited to synthetic datasets that are of the same dimension as the observed data. Indeed, there is no reason we could not generate synthetic data $z_{1:m}$, with $m\geq n$, and with corresponding empirical measure $\hat\mu_{\theta,m}$.  Such a mechanism is potentially useful to alleviate simulation noise associated with $\hat\mu_{\theta,n}$. When such a distinction between the observed and simulate size is necessary, we will write the simulated empirical measure as $\hat\mu_{\theta,m}$. 	
 
While there are many possible norms to choose from, and thus many different possible posteriors on which we can base our inferences, for reasons or efficiency and robustness we restrict our attention to classes of norms that yield ``minimum distance functionals''. Optimization of such norms is known to yield point estimators that are inherently robust and potentially efficient (see, e.g., \citealp{donoho1988automatic}). Arguably, the two most commonly used ``robust'' distances for parameter inference are the Hellinger and Cramer-von Mises (hereafter, CvM) distances.
\begin{definition}\label{def:hell}
		Let $\mu,\nu\in\mathcal{P}(\mathcal{Y})$, and let $\d\mu$ and $\d\nu$ denote the corresponding probability densities with respect to some dominating measure, say $\lambda$. The Hellinger distance between $\mu$ and $\nu$ is defined as
		$$	
		\left\{\int\left(\d\mu^{1 / 2}-\d\nu^{1 / 2}\right)^{2} \right\}^{1 / 2}=\left\{2-2 \int \d\mu^{1 / 2} \d\nu^{1 / 2} \right\}^{1 / 2}.
		$$	The Hellinger distance is equivalent to the $\mathcal{L}_2(\mathcal{P})$ norm between $\d\mu^{1 / 2}$ and $\d\nu^{1 / 2}$: $\has(\mu, \nu)=\|\d\mu^{1 / 2}-\d\nu^{1 / 2}\|_{2}$. 
	\end{definition}
	\begin{definition}\label{def:cvm}
		Let $\mu,\nu\in\mathcal{P}(\mathcal{Y})$. For integers $p,q$ and $q\leq p$, define the $\mathcal{L}_{p,q}(\mu)$ norm between $\mu$ and $\nu$, with respect to the measure $\mu$, as
		$$	
		\|\mu-\nu\|_{q,p}=\left\{\int\left(\mu-\nu\right)^{p}\d \mu \right\}^{1/q}. 
		$$For $\mu(y):=\text{Pr}(Y\leq y)$, when $\mu^{-1}$ exists, and taking $p=2$ and $q=1$, yields $$\|\mu-\nu\|_{1,2}=\left\{\int\left(\mu-\nu\right)^{2}\d \mu \right\}=\int_0^1(\nu\left[\mu^{-1}(x)\right]-x)^2\d x=\|\nu\left[\mu^{-1}(x)\right]-\cdot\|_2^2$$which is the Cramer-von Mises (CvM) distance between $\mu$ and $\nu$ (with respect to $\mu$).
	\end{definition}

\begin{remark}In Appendix \ref{app:examples}, we compare the behavior  of MD-ABC and WBC in two examples that have featured as test cases in the ABC literature; namely, the $g$-and-$k$ model and the $M/G/1$ queuing model. These results suggest that inferences based on MD-ABC compare favorably with those based on WABC. 
\end{remark}
\begin{remark}
	For continuously distributed data, the Hellinger distance requires a ``density function'' for the observed and simulated datasets. In this case, it is then beneficial to consider simulating datasets $z_{1:m}$, with $m>n$, so that the  Hellinger distance can be approximated more accurately. The observed (respectively, simulated) density function can be constructed as the convolution of $\hat{\mu}_n$ (respectively, $\hat{\mu}_{\theta,m}$) with a smooth kernel function, $\mathcal{K}_n$, that depends on bandwidth function $a_n\rightarrow0$ as $n\rightarrow\infty$, which yields the smoothed empirical measure ${\tilde\mu}_{n}(y):=(\hat{\mu}_n\star \mathcal{K}_n)(y)$ (respectively, $\tilde\mu_{\theta,m}(z):=(\hat{\mu}_{\theta,m}\star \mathcal{K}_n)(z)$), and where $f\star g$ denotes the convolution of $f$ and $g$. For inference based on the Hellinger distance, \cite{basu1994minimum} recommend that the same kernel, $\mathcal{K}_n$, is used to estimate the density of the observed data, and to construct the model-based density. When $\mu_\theta$ is tractable, no simulation is required and \cite{basu1994minimum} propose to calculate $\d\mu_\theta\star\mathcal{K}_n$ directly. 
	
\end{remark}

\begin{remark}
Our motivation for using the norm $\|\cdot\|_\mathcal{P}$ in ABC is two-fold. Firstly, in parametric models, and when the model is correctly specified, minimizing the Hellinger or CvM distances over $\mathcal{M}^{(n)}$ delivers point estimators that are as efficient as maximum likelihood estimation, or nearly so in the case of the CvM distance. Given this fact, it is highly likely that ABC based on these norms will deliver point estimators that are close to being as efficient as exact Bayes based on the intractable likelihood; we present an example Section \ref{sec:mix1} that demonstrates this property and theoretically verify this statement in Section \ref{sec:asym:postmean}. That being said, there is no reason to suspect that MD-ABC will deliver posteriors that are equivalent to those obtained by exact Bayes, given that the metrics in which the two approaches are measuring ``distance'' are different. Our second motivating factor is the well-known fact that point estimators based on the criterion $\|\hat{\mu}_n-{\mu}_{\theta}\|_\mathcal{P}$ display robustness to certain types of model misspecification. Given the results of \cite{frazier2020model} regarding the behavior or summary statistic based ABC under model misspecification, the potential robustness of MD-ABC to model misspecification could be critical in ensuring ABC delivers reliable inferences in practical applications.
\end{remark}

\subsection{Example: Mixture Model}\label{sec:mix1}
We now consider a simplistic example that allows us to easily compare MD-ABC, WABC and exact Bayes (hereafter, EB) inference. We observe a sequence $y_{1:n}=(y_1,\dots,y_n)$ generated independent and identically distributed (iid) from the following two-component mixture model: for $\omega\in[0,1]$, and $\varphi(\cdot;\mu,\sigma^2)$ denoting a normal density with mean $\mu$ and variance $\sigma^2$,
$$
f_\theta(\cdot):=(1-\omega)\varphi(\cdot;\mu,\sigma^2_1)+\omega\varphi(\cdot;-\mu,\sigma^2_2),\;\;\theta:=(\mu,\omega,\sigma_1,\sigma_2)'.
$$
From this model, we generate $n=100$ observations where the unknown parameters $\theta$ are fixed at $\theta=(-2,0.5,1,1)'$. Our prior beliefs over $\theta$ are 
$$
\mu\sim\mathcal{N}(0,1),\;\omega\sim\mathcal{U}[0,1],\;\sigma_1\sim\mathcal{U}[0,10],\;\sigma_2\sim\mathcal{U}[0,10].
$$

In this experiment, it is simple to construct the likelihood and sample from the exact posterior. To this end, EB inference is implemented using the NUTS sampler via \texttt{Rstan} (\citealp{carpenter2017stan}). {The NUTS sample was run for 10,000 iterations with an initial warmup period of 1,000 iterations.} 

We compare EB with ABC based on the Wasserstein distance (WABC), and MD-ABC using the CvM distance (CvM-ABC) and the Hellinger distance (H-ABC). We set the number of simulated sample points for CvM-ABC and WABC to be $n$, i.e., we take $m=n$, while for H-ABC we use $m=2\cdot n$. All ABC posteriors are sampled using the same approach described in the previous examples.

The resulting posteriors for EB, WABC, and MD-ABC are given in Figure \ref{fig:mixfig1}. All posteriors are centered over, or close to, the true values used to generated the data, and there is remarkable agreement across the posteriors. Figure \ref{fig:mixfig_reps1} compares the posterior means of the four different approaches across fifty replicated data sets (generated using the same specification), and demonstrates that, in repeated samples, there is substantial agreement between the point estimators from these different methods. 
\begin{figure}[H]
	\centering
	\includegraphics[height=.40\textheight,width=1.1\textwidth]{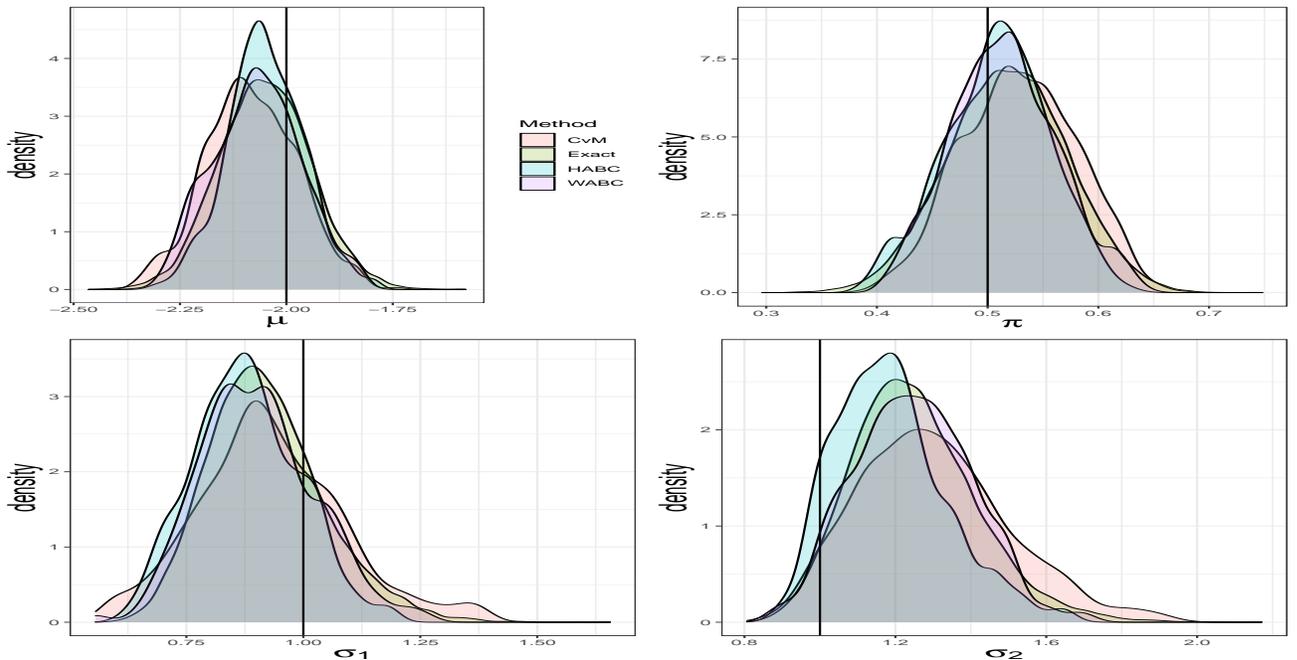}
	\caption{\small Posterior marginals in the mixture model example. The value used to generate the data is $\theta_\star=(-2,0.5,1,1)'$. Vertical lines indicate true values.}
	\label{fig:mixfig1}
\end{figure}
\begin{figure}[H]
	\centering
	\includegraphics[height=.50\textheight,width=.99\textwidth]{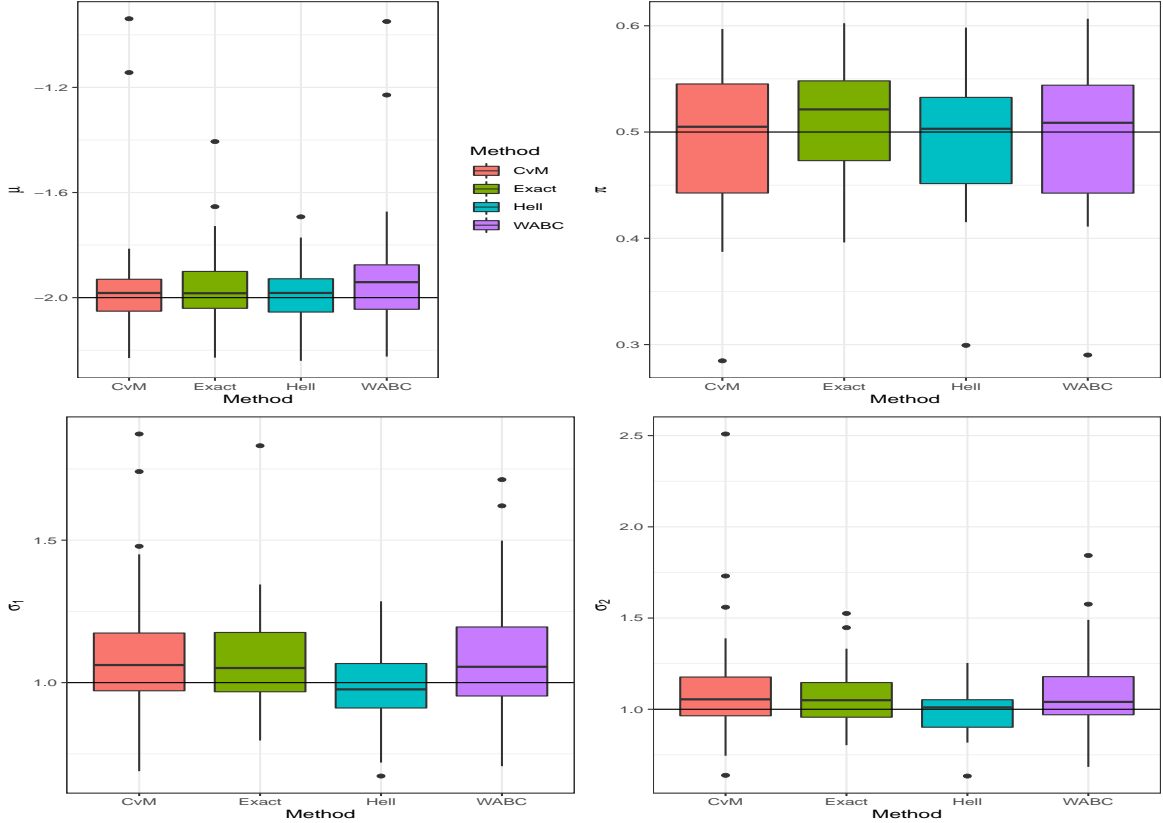}
	\caption{\small Posterior means in the mixture model example across repeated samples. The value used to generate the data is$\theta_\star=(-2,0.5,1,1)'$. Horizontal lines indicate true values.}
	\label{fig:mixfig_reps1}
\end{figure}

Lastly, Table \ref{tab:mix1} compares the posterior mean, (Means), standard deviation (Std), Monte Carlo coverage (Cov), and root mean squared error (RMSE) for the posterior mean across the replications.\footnote{Monte Carlo coverage is calculated as the percentage of the replications where the true value $\theta_\star$ is contained in the marginal 95\% credible set.} The results illustrate that all methods generally give reliable point estimates of the unknown parameters, with H-ABC displaying the smallest bias and RMSE across all parameters. In terms of posterior uncertainty, WABC displays the largest posterior uncertainty, as measured by the posterior standard deviation, while H-ABC displays the smallest amount of posterior uncertainty, followed closely by EB. In terms of Monte Carlo coverage, we note that H-ABC is arguably closest to the nominal 95\% level across the different parameters, while WABC yields 100\% coverage for all parameters except $\pi$.

The repeated sampling results demonstrate that MD-ABC can produce inferences that are competitive with those obtained using EB, and also yields point estimators that outperform WABC according to the chosen loss measures.  

\begin{table}[H]
  \centering
  \caption{Repeated sampling result for marginal posteriors. CvM refers to MD-ABC based on the CvM distance, while Hell refers to MD-ABC based on the Hellinger distance; WABC is ABC based on the Wasserstein distance, and Exact refers to exact Bayes. Means- overall posterior mean (across the replications). Std- average posterior standard deviation. Cov- Monte Carlo coverage. RMSE- root mean squared error for the posterior mean (relative to $\theta_\star=(-2,0.5,1,1)'$).}
    \begin{tabular}{rrrrrrrrrrr}
          &       & \multicolumn{1}{l}{Means} &       &       &       &       &       & \multicolumn{1}{l}{Std} &       &  \\
                & \multicolumn{1}{l}{$\mu$} & \multicolumn{1}{l}{$\pi$} & \multicolumn{1}{l}{$\sigma_1$} & \multicolumn{1}{l}{$\sigma_2$} &       &       & \multicolumn{1}{l}{$\mu$} & \multicolumn{1}{l}{$\pi$} & \multicolumn{1}{l}{$\sigma_1$} & \multicolumn{1}{l}{$\sigma_2$} \\
    {CvM} & -1.96 & 0.49 & 1.10   & 1.11  &       & {CvM} & 0.13 & 0.05 & 0.21 & 0.23 \\
    {Exact} & -1.96 & 0.51 & 1.08  & 1.07  &       & {Exact} & 0.12 & 0.06 & 0.16 & 0.15 \\
    {Hell} & -1.98 & 0.49 & 1.01 & 1.01  &       & {Hell} & 0.09& 0.05  & 0.12 & 0.12 \\
    {WABC} & -1.93 & 0.49 & 1.09 & 1.09 &       & {WABC} & 0.19 & 0.06 & 0.27 & 0.28 \\
          &       & \multicolumn{1}{l}{COV} &       &       &       &       &       & \multicolumn{1}{l}{RMSE} &       &  \\
          & \multicolumn{1}{l}{$\mu$} & \multicolumn{1}{l}{$\pi$} & \multicolumn{1}{l}{$\sigma_1$} & \multicolumn{1}{l}{$\sigma_2$} &       &       & \multicolumn{1}{l}{$\mu$} & \multicolumn{1}{l}{$\pi$} & \multicolumn{1}{l}{$\sigma_1$} & \multicolumn{1}{l}{$\sigma_2$} \\
    {CvM} & 88\%   & 92\%  & 92\%  & 96\%  &       & {CvM} & 0.21 & 0.06 & 0.24 & 0.29 \\
    {Exact} & 90\%   & 98\%  & 90\%  & 92\%  &       & {Exact} & 0.15 & 0.06 & 0.20 & 0.17 \\
    {Hell} & 92\%  & 92\%  & 90\%  & 96\%   &       & {Hell} & 0.10 & 0.06 & 0.13 & 0.12 \\
    {WABC} & 100\%   & 96\%  & 100\%  & 100\% &       & {WABC} & 0.22 & 0.06 & 0.23  & 0.22 \\
          &       &       &       &       &       &       &       &       &       &  \\
    \end{tabular}%
  \label{tab:mix1}%
\end{table}%

\section{Theoretical properties \label{sec:asymptoticproperties}}
In this section, we analyze the asymptotic behavior of MD-ABC: we prove posterior concentration, a Bernstein von-Mises result, and document the behavior of the MD-ABC posterior mean. In addition, we rigorously analyze the robustness of MD-ABC using a local perturbation scheme. The results demonstrate that MD-ABC can deliver point estimators that are as efficient as those obtained by exact Bayes,  for certain choices of $\|\cdot\|_{\mathcal{P}}$, and is simultaneously robust to deviations from the underlying modeling assumptions.  

Before rigorously analyzing the behavior of this new ABC approach, we first define the following additional notations that will be used throughout the remainder of the paper. 
For sequences $\{a_{n}\}$ and $\{b_{n}\}$, real valued, $a_{n}\lesssim b_{n}$
denotes $a_{n}\leq Cb_{n}$ for some $C>0$, $a_{n}\asymp b_{n}$ denotes equivalent order of
magnitude, $a_{n}{\gg }b_{n}$ indicates a larger order of magnitude and the
symbols $o_{p}(a_{n}),\;O_{p}(a_{n})$ have their usual meaning. We reserve $\Rightarrow $ to denote weak convergence under some stated measure. Whenever $\mu_\star$ and $\mu_\theta$ is used, it is assumed that these objects exist. Unless otherwise noted, all limits are taken as $n\rightarrow\infty$.

\subsection{Posterior Concentration\label{sec:asymptotics:double}}
For a sequence of tolerance values $\{\epsilon_n\}_{n\ge1}$, we first prove that the MD-ABC posterior concentrates onto
$
\epsilon_\star:=\inf_{\theta\in\Theta}\|\mu_{\theta}-\mu_\star\|_\mathcal{P}\geq0.
$
Since $\epsilon_\star\geq0$, this posterior concentration result will not require that the model be correctly specified. We say the model is correctly specified if there exists some $\theta\in\Theta$, possibly not unique, such that $\mu_{\theta_\star}=\mu_\star$. 

We maintain the following assumptions, which are similar to those considered in \cite{FMRR2016}:
\begin{assumption}\label{ass:1}
	There exists a positive sequence $v_n\rightarrow\infty$ such that $$\liminf_{n\rightarrow\infty}\mu^{(n)}_\star\left(\|\hat{\mu}_n- \mu_\star\|_\mathcal{P}\geq v_n^{-1} \right)=1.$$
\end{assumption}

\begin{assumption}\label{ass:2}There exist a positive sequence $v_{m}\rightarrow\infty$ as $m\rightarrow\infty$ such that for some $\kappa$ such that $d_\theta<\kappa<\infty$, and for all $u>0$, and for all $\theta \in \Theta$,  
	\begin{equation*}
	\mu^{(n)}_{{\theta }}\left[\|\hat\mu_{\theta,m}-\mu_\theta\|_\mathcal{P} >u\right] \leq c({\theta }%
	)/v_m^{-\kappa}u^{-\kappa},\quad \int_{\Theta }c({\theta })\d \Pi ({%
		\theta })<+\infty.
	\end{equation*}
\end{assumption}

\begin{assumption}\label{ass:4} There exist some $D>0$ and $M_0, \delta_0>0$ such that, for all $\delta_0\geq \delta >0$ and  $M\geq M_0$,  there exists $S_\delta \subset \left\{\theta \in \Theta : \| \mu_\theta-\mu_{\star}\|_\mathcal{P}-\epsilon_\star \leq \delta \right\}$ such that, for $D< \kappa $,$$\int_{S_\delta} ( 1  - { c(\theta)/M }{  })\d \Pi(\theta)  \gtrsim\delta^{D}.$$
\end{assumption}
\begin{remark}
	The above assumptions are similar to Assumptions 1-4 in \cite{FMRR2016} and \cite{frazier2020model} in the case of summary statistic ABC under correct model misspecification, and global model misspecification, respectively. However, in this case, the assumptions do not pertain the summary statistics calculated from the model but to the empirical measures calculated from the observed and simulated data. We refer to those papers for a general discussion of these assumptions. 
\end{remark}	

The proof of the following follows the same lines as Theorem 1 of \cite{FMRR2016}, and Proposition 3.2 in \cite{Bernton2017} and hence is omitted for brevity. 
\begin{theorem}[Theorem 1, \cite{frazier2020model}]\label{thm:one}
	Let the tolerance sequence be such that $\epsilon_n \rightarrow \epsilon^* $ with 
	$$\epsilon_n \geq \epsilon^*  +  M v_m^{-1}  + v_{n}^{-1}, $$ for $M$  large  enough. For any positive sequence $M_n\rightarrow\infty$, let $\delta_n\rightarrow0$ be such that
	$\delta_n \geq  M_n(\epsilon_n - \epsilon^*)$. If Assumptions \ref{ass:1}-\ref{ass:4} are satisfied, then
	\begin{equation}\label{post:conc:good}
	\Pi_{\epsilon}\left[ \|\mu_\theta-\mu_\star\|_\mathcal{P} \geq \epsilon^*+ \delta_{n} |y_{1:n} \right] = o_{p}(1),
	\end{equation}
	as soon as $\delta_n \geq M_n v_m^{-1}u_n^{-D/\kappa}  = o(1)$,	with $u_n =\epsilon_n - (\epsilon^*  +  M v_m^{-1}  + v_{n}^{-1}  )\geq 0$.
\end{theorem}	
\begin{remark}
The above result states that the MD-ABC posterior concentrates onto $\inf\|\mu_\theta-\mu_\star\|$ as $n\rightarrow\infty$ and $\epsilon_n\rightarrow\epsilon_\star$. Moreover, the results states that this concentration happens at the slower of the two rates for the observed and simulated datasets. Such a result is meaningful as nothing prohibits us from simulating more than the data size, which could allow us to obtain better estimates of $\mu_\theta$, than if we were restricted to simulating samples of length $m=n$. The MD-ABC approach with non-zero $\epsilon_\star$ can be viewed as a type of coarsened posterior, \cite{miller2015}. However, while \cite{miller2015} present certain arguments as to why coarsened posteriors can lead to robust inferences in the case of model misspecification, we formally prove the robustness of MD-ABC to deviations from the modeling assumptions in Section \ref{sec:robust}. Lastly, it is important to note that the result in Theorem \ref{thm:one} recovers the case of posterior concentration under correctly specified models when $\epsilon_\star=0$. 
\end{remark}

The result of Theorem \ref{thm:one} does not require the parameters to be identifiable, in which case posterior concentrates is onto the set $\{\theta\in\Theta:\epsilon_\star=\|\mu_\theta-\mu_\star\|_\mathcal{P}\}$. However, in many cases the model $\mathcal{M}^{(n)}$ will be identifiable and it would be convenient to state conditions under which the MD-ABC posterior concentrates onto a Dirac mass. 
\begin{assumption}\label{ass:5}
	(i) The map $\theta\mapsto\mu_\theta$ is injective; (ii) There exists $\theta_\star\in\Theta$ such that, for all $\delta >0$,
	$$\inf_{\theta\in\Theta: \|\theta-\theta_\star\| > \delta} \|\mu_\theta-\mu_\star\|_{\mathcal{P}} > \|\mu_{\theta_\star}-\mu_\star\|_{\mathcal{P}}.$$ (iii) For some $L>0$ and $\alpha>0$, for all $\theta$ in some open neighborhood of $\theta_\star$, $$\|\theta-\theta_\star\|\leq L\left[\|\mu_{\theta}-\mu_\star\|-\epsilon_\star\right]^{\alpha}.$$
\end{assumption}
\begin{corollary}\label{thm:cor1}Assume the hypotheses of Theorem \ref{thm:one} and Assumption \ref{ass:5} are satisfied. Then, as soon as $\delta_n \geq M_n v_m^{-1}u_n^{-D/\kappa}  = o(1)$,	with $u_n =\epsilon_n - (\epsilon^*  +  M v_m^{-1}  + v_{n}^{-1}  )\geq 0$,
	$$\Pi_{\epsilon}[\|\theta-\theta_\star\|>L\cdot\delta_n^{\alpha}|y_{1:n}]=o_{p}(1).$$
\end{corollary}

\subsection{Asymptotic Behavior of the Posterior and the Posterior Mean}\label{sec:asym:postmean}
Under the case of correct model specification, \cite{beran1977minimum} demonstrated that the minimum Hellinger distance estimator, the minimizer of $\has\left(\mu_\theta,\hat{\mu}_{n}\right)$, was asymptotically equivalent to the MLE under correct model specification. In addition, \cite{ozturk1997generalised} have shown that the minimum CvM distance estimator is asymptotically normal and can be nearly as efficient as MLE. In this section, we extend such results to the Bayesian paradigm in the case of intractable likelihoods. In particular, we demonstrate that under correct model specification: 1) the MD-ABC posterior is asymptotically Gaussian ; 2) the MD-ABC posterior mean is asymptotically equivalent to the minimizer of the map $\theta\mapsto \|\mu_{\theta}-\hat{\mu}_{n}\|_{\mathcal{P}}$. The latter result directly implies that,  in the class of regular models, MD-ABC based on the Hellinger distance delivers point estimators that are asymptotically as efficient as the MLE . 

In order to compare the efficiency of MD-ABC with other approaches, we explicitly consider the case of correctly specified models. In addition, we consider that MD-ABC is implemented with $m=n$ simulations. In what follows consider that $\|\cdot\|_\mathcal{P}$ can be associated to the inner product $\langle\cdot,\cdot\rangle$. 

Throughout the remainder, we employ the following explicit assumptions.  
\begin{assumption}\label{ass:5a}
 For some $\theta_\star\in\Theta$, $\|\mu_{\theta_\star}-\mu_\star\|_{\mathcal{P}}=0$ and for any $\theta\neq\theta_\star$, $\|\mu_{\theta}-\mu_\star\|_{\mathcal{P}}>0$. 
\end{assumption}

\begin{assumption}\label{ass:6}Let $t:=\sqrt{n}(\theta-\theta_\star)\in\Theta$ denote a local parameter. The map $\theta\mapsto\mu_\theta$ satisfies:
	\begin{enumerate}
		\item There exists a vector $\dot\xi=(\dot\xi_1,\dots,\dot\xi_{d_\eta})'$ such that $\|\mu_{\theta_\star+t/\sqrt{n}}-\mu_\star-\langle t,\dot\xi \rangle\|_{\mathcal{P}}=o(\|t\|)$. 
		\item The vector $\dot\xi$ is such that $H_\star:=\E[\dot\xi\dot\xi']$ is non-singular. 
		\item For any $t\neq0$, there exists constants $C_1,C_2>0$, $C_2\geq C_1$, such that $C_1\|t\|\leq\|\langle t,\dot\xi\rangle\|\leq C_2\|t\|$
	\end{enumerate}
\end{assumption}

\begin{assumption}\label{ass:seven}(i) There exists some $\delta>0$ such that for all $\theta\in\Theta$ with $\|\theta-\theta_\star\|\leq\delta$, for $\mathbb{G}_{\theta}$ a mean-zero Gaussian process with covariance kernel $\mathbb{E}\mathbb{G}_{\theta}(z_1)\mathbb{G}_{\theta}(z_2):=\mu_\theta(z_1\wedge z_2)-\mu_\theta(z_1)\mu_\theta(z_2)$, 
$$
\sqrt{n}\left[\hat{\mu}_{n,\theta}-\mu_\theta\right]\Rightarrow\mathbb{G}_{\theta},\text{ and }\sqrt{n}\left[\hat{\mu}_{n}-\mu_{\theta_\star}\right]\Rightarrow\mathbb{G}_{\theta_\star}.
$$ 	
\end{assumption}

\begin{assumption}\label{ass:eight} For some positive $\delta $, all $\Vert {\theta }-%
	{\theta }_{\star}\Vert \leq \delta $, and for all sets $B_{n}=\big\{%
	t:\|t\|\le h_n%
	\big\}$ and all $u\in \mathbb{R}$ fixed, for some $%
	h_{n}\rightarrow 0$ as $n\rightarrow +\infty $, for some $V_\theta>0$, 
	\begin{equation}
	\begin{split}
	\lim_{n\rightarrow\infty}
	\frac{1}{h_{n}}\,{\mu^{(n)}_{{\theta }}\left[ \{\hat{\mu }_{\theta}-\mu_\theta\}-u\in B_{n}\right] }
	& =\varphi(u;0,{V_\theta}), \\
	\frac{1}{h_{n}}\,{\mu^{(n)}_{{\theta }}\left[ \{\hat{\mu }_{\theta}-\mu_\theta\}-u\in B_{n}\right] } & \leq H(u),\quad \int
	H(u)du<+\infty ,
	\end{split}
	\label{dens:cond}
	\end{equation}%
	for $\varphi(\cdot;0,{V_\theta})$ the density of a mean-zero normal random variable with variance $V_\theta$.\end{assumption}

\begin{assumption}\label{ass:nine} The prior density $\pi$ exists and is such that:
	{(i)} For $\theta _{\star}\in \text{Int}({\Theta })$, $%
	\pi(\theta _{\star})>0$; {(ii)} For some $\delta>0$, the density function 
	$\pi(\cdot)$ is continuously differentiable for all $\|\theta -\theta
	_{\star}\|\leq \delta $; {(iii)} For $\Theta \subset \mathbb{R}^{d_\theta}$, and some $\beta>d_\theta$, we have $\int_{\Theta
	}\|\theta \|^{\beta }\pi(\theta )d\theta <+\infty $.
\end{assumption}

\begin{remark}
	Assumptions \ref{ass:5a} requires that the model is correctly specified but only in terms of the norm $\|\cdot\|_{\mathcal{P}}$. That is, if the model is not correctly specified, then it is entirely possible that the value of $\theta_\star$ under which Assumptions \ref{ass:5a} is satisfied may differ across choices of $\|\cdot\|_{\mathcal{P}}$. Assumption \ref{ass:6} implicitly imposes smoothness on the map $\theta\mapsto\mu_\theta$ in relation to the chosen norm, and is often referred to as norm differentiability. Assumptions \ref{ass:seven} specifies that the empirical measure satisfies a uniform CLT in a neighborhood of $\theta_\star$. Assumption \ref{ass:eight} is similar to Assumption 8 in \cite{FMRR2016} and gives the requisite regularity needed to deduce asymptotic normality of the posterior.
\end{remark}
Define the following point estimator of $\theta_\star$
\begin{flalign}\label{eq:md1}
\hat{\theta}_n&:=\argmin_{\theta\in\Theta}\|\mu_{\theta}-\hat\mu_n\|_\mathcal{P},
\end{flalign}where we highlight that $\mu_\theta$ in \eqref{eq:md1} is not estimated via simulation but is known. As such, in many models, there is no reason to suspect that $\hat\theta_n$ is feasible to obtain in practice. 
\begin{theorem}
	\label{normal_thm} Under Assumptions \ref{ass:1}-\ref{ass:nine},  if $\sqrt{n}\epsilon _{n}=o(1)$, for any $A\in\Theta$,
	\begin{equation}
\lim_{n\rightarrow+\infty}\Pi _{\epsilon }\left[ \sqrt{n}( \theta-\hat\theta_n)\in A|y_{1:n}\right] =\int_{A}\varphi(x;0,V_{\theta_\star})\d x,  \label{anennorm}
\end{equation}
\end{theorem}

\begin{corollary}	\label{mean_thm} If Assumptions \ref{ass:1}-\ref{ass:nine} are satisfied, and if $\sqrt{n}\epsilon_n=o(1)$ then, for $\bar\theta_n=\int\theta\d\Pi _{\epsilon }(\theta|y_{1:n})$, 
	\begin{flalign}
\sqrt{n}\|\bar\theta_n-\hat\theta _n\|=o_p(1),\label{normality:postmean}
	\end{flalign}
	%
\end{corollary}
\begin{remark}
Theorem \ref{mean_thm} implies that the MD-ABC posterior mean is asymptotically equivalent to the minimum distance estimator in \eqref{eq:md1}, even in cases where $\hat\theta_n$ can not be feasibly computed. Consequently, under correct model specification, the MD-ABC posterior mean based on the Hellinger distance is an asymptotically efficient estimator of $\theta_\star$ (since the MHD estimator is itself asymptotically equivalent to the maximum likelihood estimator, see, e.g, \citealp{beran1977minimum}). Therefore, the MD-ABC posterior mean, under the Hellinger distance, is asymptotically as efficient as the MLE. Given the asymptotic equivalence between the MLE and the exact Bayes posterior mean, satisfied under regularity, we can conclude that the MD-ABC posterior mean is asymptotically as efficient as that obtained under exact Bayesian inference (even when the latter is infeasible to obtain).
\end{remark}
\begin{remark}
	While MD-ABC and WABC are both based on the entire sample of observed data, $y_{1:n}$, it is unclear whether point estimators obtained from WABC satisfy a result similar to that in Theorem \ref{mean_thm}. This makes directly comparing the behavior of such estimators difficult, and so throughout any comparison with WABC must proceed via simulation. 
\end{remark}

\begin{remark}
	Under similar conditions to those maintained herein, \cite{pollard1980minimum} demonstrates that $\sqrt{n}(\hat\theta_n-\theta_\star)\Rightarrow \inf_{t\in\mathbb{R}^{d_\theta}}\|\mathbb{G}_{\theta_\star}-\langle t,\dot\xi\rangle\|_\mathcal{P}$. When this minimum is unique, and for $\mathbb{G}_{\theta_\star}$ denoting a Gaussian process, the limiting distribution will often have a standard asymptotically linear form. For instance, if $\|\cdot\|_\mathcal{P}$ is taken to be the Hellinger or CvM distances, this minimum will be unique (i.e., under the identification condition in Assumption \ref{ass:5a}). That being said, there is generally no guarantee that this minimum will be unique and in such case all we can affirm is that  $\sqrt{n}(\hat\theta_n-\theta_\star)$, and hence $\sqrt{n}(\bar\theta_n-\theta_\star)$, will behave like one of these minimizers.
\end{remark}

\subsection{Robustness}\label{sec:robust}
Recall that the dual objective of MD-ABC was to deduce a method that was more efficient than summary-statistic based ABC, but which simultaneously displays meaningful robustness to deviations from the underlying modeling assumption. Corollary \ref{mean_thm} demonstrates that MD-ABC satisfies the first of these tenants, and in this section we formally demonstrate that MD-ABC point estimators display robustness to perturbations from the assumed model. Throughout the remainder, we refer to such a setting as local model misspecification.\footnote{We refer the reader to \cite{frazier2020model} for an analysis of ABC methods under global model misspecification.}

Let 
 $T=\{T_n\}_{n\ge1}$, $T_n:\mathcal{M}^{(n)}\rightarrow\Theta$ denote a functional of interest. For all $Q\in\mathcal{M}^{(n)}$, define the map $Q\mapsto \bar{T}(Q)$ by 
$$
\bar{T}(Q):=\argmin_{t\in\Theta} \int \frac{(\sqrt{n}\{t-\theta\})\d\Pi(\theta)\int_{\mathcal{Y}^n}\mathds{1}\left[\|\hat{\mu}_{\theta,n}-Q\|_\mathcal{P}\leq\epsilon\right]\d\mu_\theta^{(n)}(z_{1:n})}{\int_{{\Theta}}\d\Pi(\theta)\int_{\mathcal{Y}^n}\mathds{1}\left[\|\hat{\mu}_{\theta,n}-Q\|_\mathcal{P}\leq\epsilon\right]\d\mu_\theta^{(n)}(z_{1:n})}.
$$
The MD-ABC posterior mean corresponds to $\bar{\theta}_n=\bar{T}_{}(\hat{\mu}_n)$. 

 We now analyze the behavior of the map $Q\mapsto \bar{T}_{}(Q)$ as we vary the true process that has generated the data in a neighborhood of the assumed model. To this end, following, among others, \cite{rieder2012robust}, we assume that the observed data is generated from a member in the class of \textit{simply perturbation} from the assumed model $\mu^{(n)}_\theta$ along the tangent $\zeta_n$: for and $t\in\Theta$, $\theta_n:=\theta_\star+t/\sqrt{n}$, observed data is generated according to\footnote{Similar to the analysis in Section \ref{sec:asym:postmean}, in this section we only consider the specific case where $v_n=v_m=\sqrt{n}$ to simplify the analysis.}
 $$
 \d \mu^{(n)}_{\theta_n,\zeta_n}=\left(1+\frac{t^{\intercal}\zeta_n}{\sqrt{n}}\right)\d\mu^{(n)}_{\theta_\star},\quad\lim _{n \rightarrow\infty} \int(\zeta_n-\zeta)^2\d\mu_\theta=0, \quad \left|\zeta_{n}\right|=o(\sqrt{n}). 
 $$It is important to note that the form of model misspecification induced by the simple perturbations $\mu_{\theta_n,\zeta_n}$ excludes the case of global model misspecification analyzed in \cite{frazier2020model}. Unlike the clase of global model misspecification, the ``modelling bias'' associated with the simple perturbations eventually converges to zero.

 Define the neighborhood of radius $r/\sqrt{n}$ around $\mu_\star$, based on the distance $d(\cdot,\cdot)$, as follows:
 $$\mathcal{B}_{d}\left(\mu_\star, r / \sqrt{n}\right)=\left\{\mu \in \mathcal{M}^{}: d\left(\mu^{}, \mu_\star^{}\right) \leq r / v_n\right\}.$$ Implicitly, the neighborhood $\mathcal{B}_{d}\left(\mu_\star, r / \sqrt{n}\right) $ sets the classes of models under which we will analyze the behavior of functionals $Q\mapsto T(Q)$ and statistics obtained from $T(\cdot)$. Optimality of these functionals is defined relative to a given neighborhood system in which we measure the deviations from the model. 
 
 Herein, we consider two common types of neighborhoods: Hellinger and Cramer-von Mises (CvM) neighborhoods. The Hellinger neighborhood replaces the generic distance $d$ in $\mathcal{B}_d$ with the Hellinger distance, see Definition \ref{def:hell}, and we denote the neighborhood by $\mathcal{B}_{\has}$. To define the CvM neighborhood, let $\lambda$ be a $\sigma$-finite measure,  and let $Q(y)$ and $P(y)$ denote the CDF of $Q$ and $P$ at the point $y\in\mathcal{Y}$. Then, the CvM neighbourhood is defined as 
 $$
\mathcal{B}_{\mathcal{C}}\left(\mu_\star, r / \sqrt{n}\right),\text{ where } \mathcal{C}(Q,P):=\int_{\mathcal{Y}} \left|Q(y)-P(y)\right|\lambda(\d y).
 $$

We now analyze the robustness of MD-ABC to deviations from the model that lie in $\mathcal{B}_{d}$, for $d\in\{\has,\mathcal{C}\}$, and compare this behavior to any other functional $\mu\mapsto T_\eta(\mu)$ computed from summary statistics $\eta$. We carry out this analysis in the case where the functionals are regular and Fisher consistent.  Let $T_a\in\{\bar{T},T_\eta\}$, $T_a:\mathcal{M}\mapsto\Theta$ be an estimator of $\theta_\star$. 

\begin{definition}\label{def:pert}
	Let $\mu_{\theta,\zeta_n}$ denote a sequence of \textit{simple perturbations} of $\mu_\theta$ along the tangent $\zeta_n$ such that: $\mu_{\theta,0}=\mu_\theta$ and, eventually, for $t\in\Theta$ such that $\mu_{\theta+t/\sqrt{n},\zeta_n}\in\mathcal{B}_{d}(\mu_\theta,r/\sqrt{n})$.
	\begin{itemize}
		\item[(i)] $T_a$ is \textit{Fisher consistent} if
		$
		\sqrt{n}\left[T_a\left(\mu_{\theta_\star+t/\sqrt{n},\zeta_n}\right)-\theta\right]=t + o(1).
		$ The function $\varphi_\theta$ is called the \textit{influence curve} of $T_a$ at $\mu_\theta$. 
		\item[(ii)] $T_a$ is called \textit{regular} for $\theta$ if for every $(\mu_{\theta+t/\sqrt{n},\zeta_n})_{n\ge1}$, every $t\in\Theta$, there exists some measure $M$, not depending on the sequence $\{(\theta+t/\sqrt{n})',\zeta_n')'\}_{n\ge1}$ such that 
		$$
		\sqrt{n}\left[T_a(\hat\mu_n)-T_a\left(\mu_{\theta+t/\sqrt{n},\zeta_n}\right)\right]\Rightarrow M,\text{ under }\mu_{\theta+t/\sqrt{n},\zeta_n}.
		$$
	\end{itemize}	
\end{definition}

Let $\tau:\Theta\mapsto\mathbb{R}$ be a (possibly) nonlinear transformation of the parameter with a bounded first derivative.\footnote{The use of $\tau$ simplifies the calculation of bias and variance terms in the following calculations.} The following result describes the robustness of MD-ABC functionals to deviations from the underlying modelling assumptions. 
\begin{theorem}\label{thm:bias}For each $r>0$, $d\in\{\mathcal{C},\has\}$, and for $L_d^*=(\d\tau(\theta_\star)/\d\theta)H^{-1}_{\star,d}(\d\tau(\theta_\star)/\d\theta)'$, where $H_{\star,d}$ is defined as in Assumption \ref{ass:6},
	\begin{itemize}
		\item[(i)]For all functionals, $T_a$, satisfying Definition \ref{def:pert}, $$\liminf _{n \rightarrow \infty} \sup _{Q \in \mathcal{B}_{\has}\left(\mu_{\star}, r / \sqrt{n}\right)} n\left\{\tau \circ T_{a}(Q)-\tau\left(\theta_{0}\right)\right\}^{2} \geq 4 r^{2} L^{*}_\has,$$  and $$\liminf _{n \rightarrow \infty} \sup _{Q \in \mathcal{B}_{\mathcal{C}}\left(\mu_{\star}, r / \sqrt{n}\right)} n\left\{\tau \circ T_{a}(Q)-\tau\left(\theta_{0}\right)\right\}^{2} \geq r^{2} L^{*}_\mathcal{C}.$$ 
		\item[(ii)] If $\|\cdot\|_{\mathcal{P}}$ is the Hellinger distance, $\lim _{n \rightarrow \infty} \sup _{Q \in B_{\has}\left(\mu_{\star}, r / \sqrt{n}\right)} n\left\{\tau \circ \bar{T}_{}(Q)-\tau\left(\theta_{0}\right)\right\}^{2}=4 r^{2} B^{*}_\has$.
		\item[(iii)] If $\|\cdot\|_{\mathcal{P}}$ is the CvM distance,
		$\lim_{n \rightarrow \infty} \sup _{Q \in B_{\mathcal{C}}\left(\mu_{\star}, r / \sqrt{n}\right)} n\left\{\tau \circ \bar{T}_{}(Q)-\tau\left(\theta_{0}\right)\right\}^{2}=r^{2} B^{*}_{\mathcal{C}}$
	\end{itemize}
\end{theorem}
\begin{remark}
	Theorem \ref{thm:bias} demonstrates that the functional $\bar{T}_{}(\cdot)$ has the smallest asymptotic bias to deviations within a Hellinger neighborhood (resp., CvM neighborhood) when we take $\|\cdot\|_{\mathcal{P}}$ to be the Hellinger metric (resp., CvM distance). This result implies that if one is conducting inference via ABC, no choice of summary statistics, or other metric if using MD-ABC, will yield a functional that has smaller bias (to deviations within this neighborhood). This result compliments existing result on the robustness of minimum distance functional (see, e.g., \citealp{rieder2012robust}, \citealp{donoho1988automatic}) to the case where the measure is calculated using simulated data. Here, bias corresponds to the deviation of $Q$ from $\mu_\star$, where this deviation is confined to lie in a shrinking Hellinger (resp., CvM) neighborhood of $\mu_\star$.
\end{remark}

\begin{remark}
	Theorem \ref{thm:bias} demonstrates that the MD-ABC functional that results from choosing the Hellinger metric delivers optimally robust point estimators, in terms of bias, to deviations from the model (that are restricted to lie in a Hellinger neighborhood). Furthermore, we know that MD-ABC based on the Hellinger metric achieves the smallest possible asymptotic variance for regular estimators (by Theorem \ref{mean_thm}). Taken together, this implies that MD-ABC based on the Hellinger metric is not only robust to model misspecification, but will be as efficient as possible under correct model specification. 
\end{remark}
\begin{remark}
	We recall that the simple perturbations $\mu_{\theta_n,\zeta_n}$ excludes the case of global model misspecification. Under global model misspecification, the ``modelling bias'' does not converge to zero and such an approach is therefore ill-suited to determine how minor or moderate deviations from the model assumptions, such as for instance under minor data contamination, affect the resulting behavior of the point estimators. Therefore, we regard these results as a compliment to the posterior concentration result given in  \ref{thm:one} and the general results in \cite{frazier2020model}. Additional work is required to determine if the resulting MD-ABC approach is also robust to global model misspecification.
\end{remark}

\section{Additional Examples: Robustness\label{sec:robust_ex}}
The results in Section \ref{sec:robust} suggest that the MD-ABC approaches have useful robustness properties that may not be replicated by alternative inference approaches. In this section, we investigate this claim by comparing the robustness of MD-ABC, WABC and exact Bayes to deviations from the underlying modeling assumptions.

\subsection{Example: Mixture Model Revisited}We now revisit the mixture model to analyze the robustness of the various inference approaches in this context. The assumed data generating process (DGP) for $z_1,\dots,z_n$ is independent and identically distributed (iid) from a two-component mixture: for $\omega\in[0,1]$, recalling that $\varphi(\cdot;\mu,\sigma^2)$ denotes a normal density with mean $\mu$ and variance $\sigma^2$, 
$$
f_\theta(\cdot):=(1-\omega)\varphi(\cdot;\mu,\sigma^2_1)+\omega\varphi(\cdot;-\mu,\sigma^2_2),\;\;\theta:=(\mu,\omega,\sigma^2_1,\sigma^2_2)'.
$$Following \cite{cutler1996minimum}, we examine the situation where the observed data $y_1,\dots,y_n$ is generated iid from the contaminated normal mixture: for $\alpha\in[0,1]$,
$$
f_\star(\cdot):=(1-\alpha)f_\theta(\cdot) + \alpha \varphi(\cdot;\zeta,\nu).
$$The value $\alpha$ determines the level of contamination in the observed data and for $\alpha>0$ the assumed DGP is misspecified. The impact of this assumption on the resulting inferences can be controlled by choosing $\alpha,\zeta,\nu$. 

In this example we compare the robustness of exact Bayesian (EB) inference, ABC based on the Wasserstein distance (WABC), and MD-ABC based on the CvM and Hellinger distances. We fix $\alpha=0.05$, $\nu=0.01$, and we consider a grid of values for $\zeta\in[-10,10]$. The parameters $\theta$ are fixed at $\theta_\star=(-2,0.5,1,1)'$. The observed sample size is again $n=100$, and we set the number of simulated sample points for CvM-ABC to be $n$, i.e., we take $m=n$, and for H-ABC we set $m=2\cdot n$.

Across all experiments, EB inference is implemented using the NUTS sampler: {NUTS was run for 10,000 iterations with an initial warm-up period of 1,000 iterations.} For all ABC approaches, we again target the posterior using the sequential Monte Carlo
(SMC) approach of \citep{del2012adaptive} with $N=1,024$ particles, and with all procedures using $2\cdot10^5$ simulations from the model.

First, we analyze the robustness properties of these approaches by comparing the posteriors when the observed data is simulated using a particularly large value of $\zeta$, namely $\zeta=9$. Figure \ref{fig:mixfig2} displays the posteriors for the four procedures in this example. In this case, we see that the ABC posteriors are similar to those obtained under correct specification (see Figure \ref{fig:mixfig1}) and signal that these methods are somewhat insensitive to the contamination in the data. In contrast, the contamination leads to significantly different outcomes for EB: for all parameters except $\sigma_2$, the EB posteriors have fatter tails and are farther away from the true values used to generate the data than for the ABC posteriors. 
\begin{figure}[H]
	\centering
	\includegraphics[height=.45\textheight,width=1.1\textwidth]{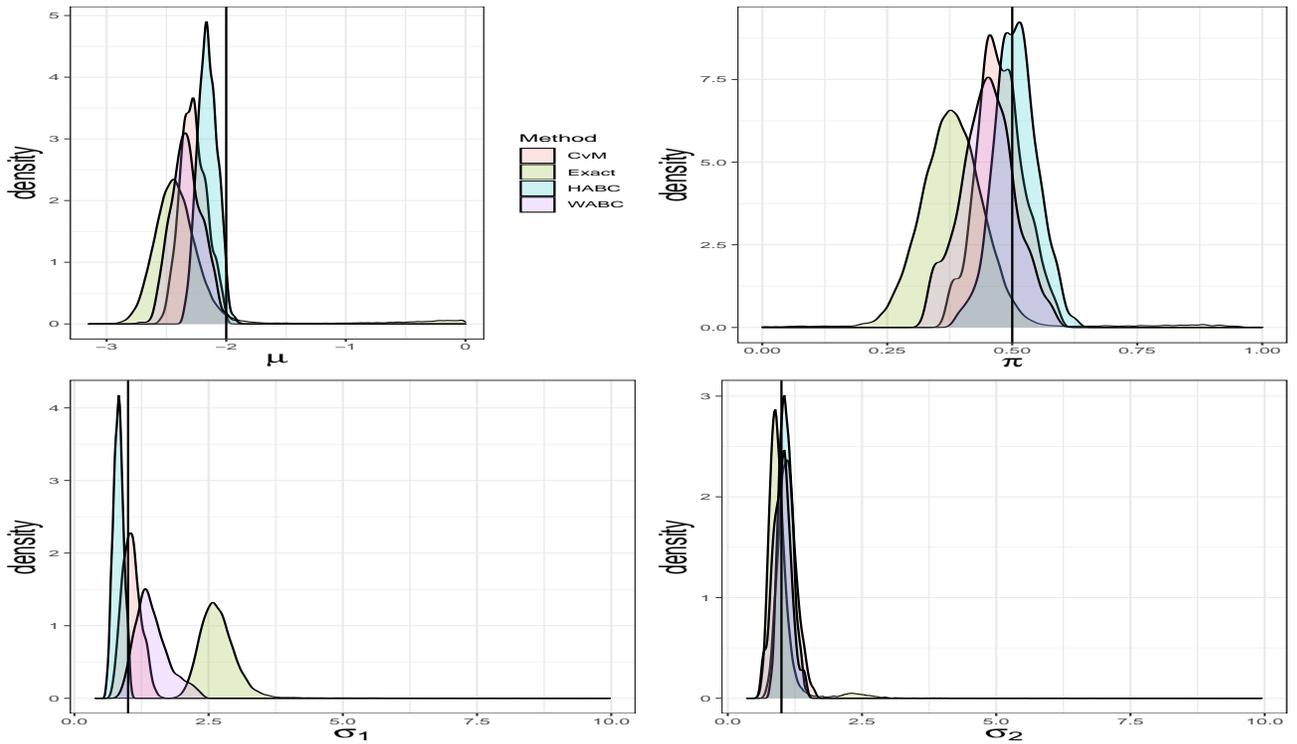}
	\caption{\small Posterior marginals in the mixture model example with data contamination. The values used to generate the data are $\theta_\star=(-2,1,1,0.5)'$ with $\alpha=0.05$ and $\zeta=9$, where the $\theta_\star$ values are denoted by vertical lines. The standard diagnostics tools in \texttt{rstan} detected no significant convergence issues.}
	\label{fig:mixfig2}
\end{figure}

In Figure \ref{fig:mixfig3} we graphically compare the evolution of the posteriors means for these four approaches across a grid of values $\zeta\in[-10,10]$, with incitements of $0.50$. The results suggest that EB can be highly sensitive to the data contamination, while MD-ABC and WABC display significantly less sensitivity to model misspecification. Across the different values of $\zeta$, and across the different methods, MD-ABC based on the CvM distance displays the least amount of sensitivity to model misspecification across the parameters. 

\begin{figure}[H]
	\centering
	\includegraphics[height=.6\textheight,width=1.15\textwidth]{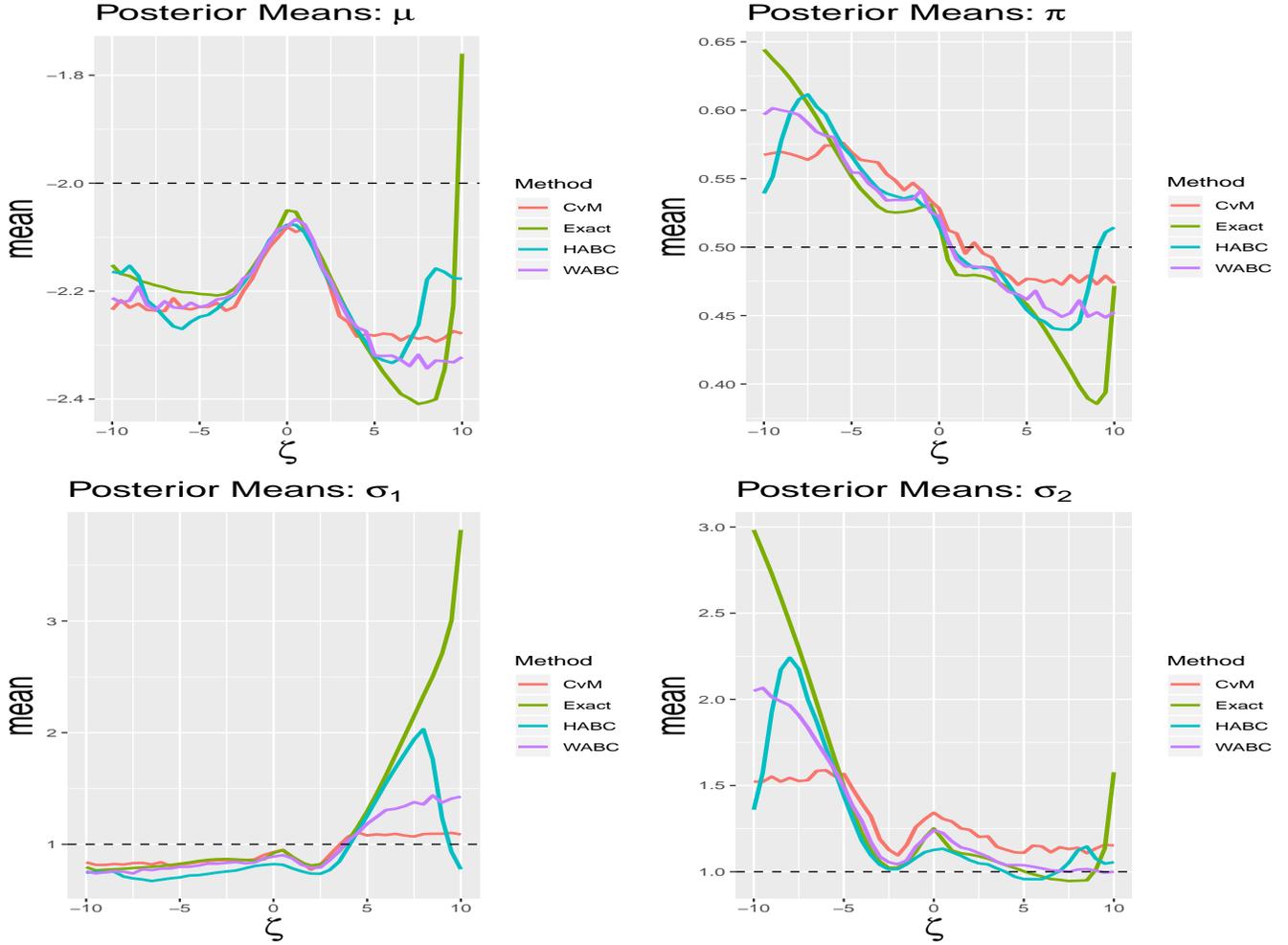}
	\caption{\small Posterior marginal means in the mixture model example with data contamination. The values used to generate the data are $\theta_\star=(-2,0.5,1,1)'$ with $\alpha=0.05$ and across the grid $\zeta\in[-10,10]$. The standard diagnostics tools in \texttt{rstan} detected no convergence issues. The horizontal dashed lines represent the true value of the parameters used in the simulation. }
	\label{fig:mixfig3}
\end{figure}

Figure \ref{fig:mixfig_rep2} displays the marginal posterior means of the four inference approaches across fifty data sets generated from the above design with $\zeta=9$. Collectively, the results in Figures \ref{fig:mixfig2}-\ref{fig:mixfig_rep2} demonstrate that EB is the most sensitive of the methods to this particular form of model misspecification, while MD-ABC based on the CvM distance displays the least amount of sensitive to misspecification, followed closely by MD-ABC based on the Hellinger distance. 

\begin{figure}[H]
	\centering
	\includegraphics[height=.55\textheight,width=.99\textwidth]{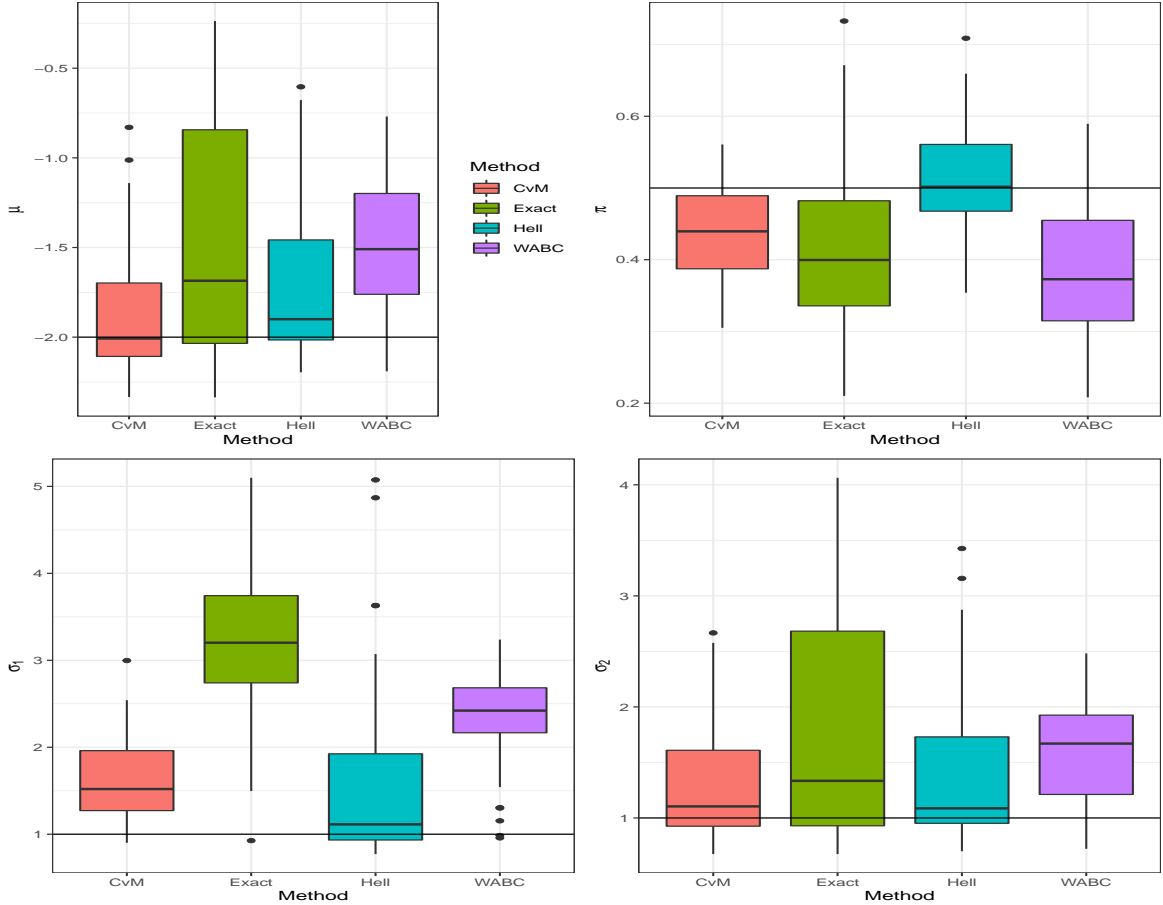}
	\caption{\small Repeated sample posterior marginals means in the mixture model example with data contamination. The values used to generate the data are $\theta_\star=(-2,1,1,0.5)'$ with $\alpha=0.05$ and $\zeta=9$. The standard diagnostics tools in \texttt{rstan} detected no significant convergence issues.}
	\label{fig:mixfig_rep2}
\end{figure}

Table \ref{tab:mix2} compares the posterior means (Means), standard deviation (Std), root mean squared error (RMSE) of the posterior mean and the Monte Carlo coverage (COV) from the replications associated with the results in Figure \ref{fig:mixfig_rep2}. As suggested by Figures \ref{fig:mixfig3} and \ref{fig:mixfig_rep2}, MD-ABC based on the CvM and Hellinger distances displays the least amount of sensitivity to model misspecification, while EB is the most sensitive. Indeed, the posterior means of EB are generally far away from $\theta_\star$, and its standard deviation is generally quite large. Moreover, due to the large bias, EB displays very poor coverage for $\sigma_1$, with the 95\% credible set covering the true value in only 2\% of the repeated samples. In contrast to EB, the marginal MD-ABC posteriors are significantly less sensitive to the model misspecification, with MD-ABC displaying better behavior than EB or WABC across virtually all the accuracy measures.

The repeated sampling results demonstrate that while EB is the gold standard under correct model specification, robust Bayesian approaches should be entertained when model misspecification is a possibility. The robustness of MD-ABC in this example, relative to EB and WABC, is a direct consequence of the methods focus on distance metrics which are known to yield inferences that are less sensitive to perturbations from the underlying modeling assumptions. 

\begin{table}[htbp]
	\centering
	\caption{Repeated sampling result for marginal posteriors in the misspecified mixture model for $\zeta=9.0$. See Table \ref{tab:mix1} for further details.}
	\begin{tabular}{rrrrrrrrrrr}
		&       & \multicolumn{1}{l}{Means} &       &       &       &       &       & \multicolumn{1}{l}{Std} &       &  \\
		& \multicolumn{1}{l}{$\mu$} & \multicolumn{1}{l}{$\pi$} & \multicolumn{1}{l}{$\sigma_1$} & \multicolumn{1}{l}{$\sigma_2$} &       &       & \multicolumn{1}{l}{$\mu$} & \multicolumn{1}{l}{$\pi$} & \multicolumn{1}{l}{$\sigma_1$} & \multicolumn{1}{l}{$\sigma_2$} \\
		CvM   & -1.88 & 0.43 & 1.61  & 1.29   &       & CvM   & 0.30 & 0.08 & 0.42 & 0.54 \\
		Exact & -1.48 & 0.42 & 3.19  & 1.76  &       & Exact & 0.45 & 0.16 & 0.98 & 1.11 \\
		Hell  & -1.73 & 0.50 & 1.58  & 1.37  &       & Hell  & 0.28 & 0.10 & 0.69 & 0.62 \\
		WABC  & -1.49 & 0.39 & 2.31  & 1.61  &       & WABC  & 0.56 & 0.13 & 0.64 & 0.93 \\
		&       & \multicolumn{1}{l}{COV} &       &       &       &       &       & \multicolumn{1}{l}{RMSE} &       &  \\
		& \multicolumn{1}{l}{$\mu$} & \multicolumn{1}{l}{$\pi$} & \multicolumn{1}{l}{$\sigma_1$} & \multicolumn{1}{l}{$\sigma_2$} &       &       & \multicolumn{1}{l}{$\mu$} & \multicolumn{1}{l}{$\pi$} & \multicolumn{1}{l}{$\sigma_1$} & \multicolumn{1}{l}{$\sigma_2$} \\
		CvM   & 88\%  & 86\% & 58\%  & 96\%  &       & CvM   & 0.37  & 0.10  & 0.77  & 0.58 \\
		Exact & 76\%  & 78\%  & 02\%  & 84\%   &       & Exact & 0.85  & 0.14 & 2.34  & 1.32 \\
		Hell  & 84\% & 94\%  & 88\%  & 92\%   &       & Hell  & 0.52  & 0.07 & 1.18   & 0.79 \\
		WABC  & 100\%  & 82\%  & 34\%  & 100\%  &       & WABC  & 0.63  & 0.14  & 1.42   & 0.79 \\
		&       &       &       &       &       &       &       &       &       &  \\
	\end{tabular}%
	\label{tab:mix2}%
\end{table}%

\subsection{Example: Stochastic Volatility}

In this example, we analyze the stability of MD-ABC inference in a time series context where data is generated according to an autoregressive log-normal stochastic volatility model. The observed series $y_{t}$ is generated according to 
$$
\begin{array}{l}{y_{t}=\sigma_{t} v_t} \\ {\log \sigma_{t}^{2}=\pi+\beta \log \sigma_{t-1}^{2}+\sigma_{u} u_{t}, \quad t=1, \ldots, T},\end{array}
$$where $v_t,u_t$ are, for simplicity, $\mathcal{N}(0,1)$ iid random variables. The unknown parameters $\theta=(\pi,\beta,\sigma_u)'$ have the following prior specification
$$
\pi\sim\mathcal{N}(0,1),\delta\sim\mathcal{U}[0,1],\sigma_u\sim\mathcal{U}[0,5].
$$

To analyze the stability of MD-ABC to local deviations from the SV model, we consider a sequence of models where the error process $v_t$ is generated from the normal mixture:
$$
f_{v_t}(\cdot):=(1-\alpha)\varphi(\cdot;0,1)+\alpha\varphi(\cdot;\zeta,.001),\text{ for }\alpha\in[0,1].
$$This particular error process emulates the ``jumps'' that are observed in financial return series, where the parameter $\alpha$ proxies as ``jump intensity/frequency'' and the parameter $\zeta$ proxies for ``jump magnitude''. We generate realizations from the above process with $\alpha=0.05$,  and where $\zeta$ varies on a uniform grid from -5 to 0, with increments of 0.50. Following the simulation design of \cite{jacquier2002bayesian}, we fix $\theta=(-0.736,0.900,0.363)'$. The sample size is set to $n$ = 500 across all experiments. 

To assure convergence to the stationary distribution, the return sequence is generated using a sample size of 1000, where the first 500 observations are discarded. The same procedure is carried out when simulating samples for ABC. 

In this example, we compare the posterior means and standard deviations of MD-ABC, based on the CvM and Hellinger distance, against those from the exact Bayes (EB), and WABC. The exact posterior is estimated using the NUTS sampler (with the same specification as in Example \ref{sec:mix1}). For MD-ABC based on the Hellinger distance, we use $m=10\cdot n$ simulated observations, and for MD-ABC based on the CvM distance, we use $m=n$. All ABC approaches used the sequential Monte Carlo
(SMC) approach of \citep{del2012adaptive} with $N=1,024$ particles, and $2\cdot10^5$ simulations from the model. 

Figures \ref{fig:svfig1} and \ref{fig:svfig2} present the results for the posterior means and standard deviations, respectively, across the difference values of $\zeta$. The results demonstrate that there is reasonable agreement across all methods when $\zeta$ is small, and indicates that MD-ABC can be reliably applied to time-series settings. However, for smaller values of $\zeta$, i.e., more negative values, the posterior means from EB are severely affected by the jump components and depart significantly from the true values used to generate the data. In addition, the posterior standard deviations for EB fluctuate significantly over the grid of values for $\zeta$. In contrast, MD-ABC point estimates of $\pi$ and $\beta$, as well as their posterior standard deviations, are virtually unaffected by the jumps in the data; inference on $\sigma$ is also less sensitive but does display some sensitivity. These results reinforce that the fact that even under relatively mild model misspecification, the reliability of EB-based inference can rapidly degraded; even when only five percent of the observed data does not agree with the assumed model, as in this example. 

\begin{figure}[H]
	\centering
	\includegraphics[height=.75\textheight,width=.75\textwidth]{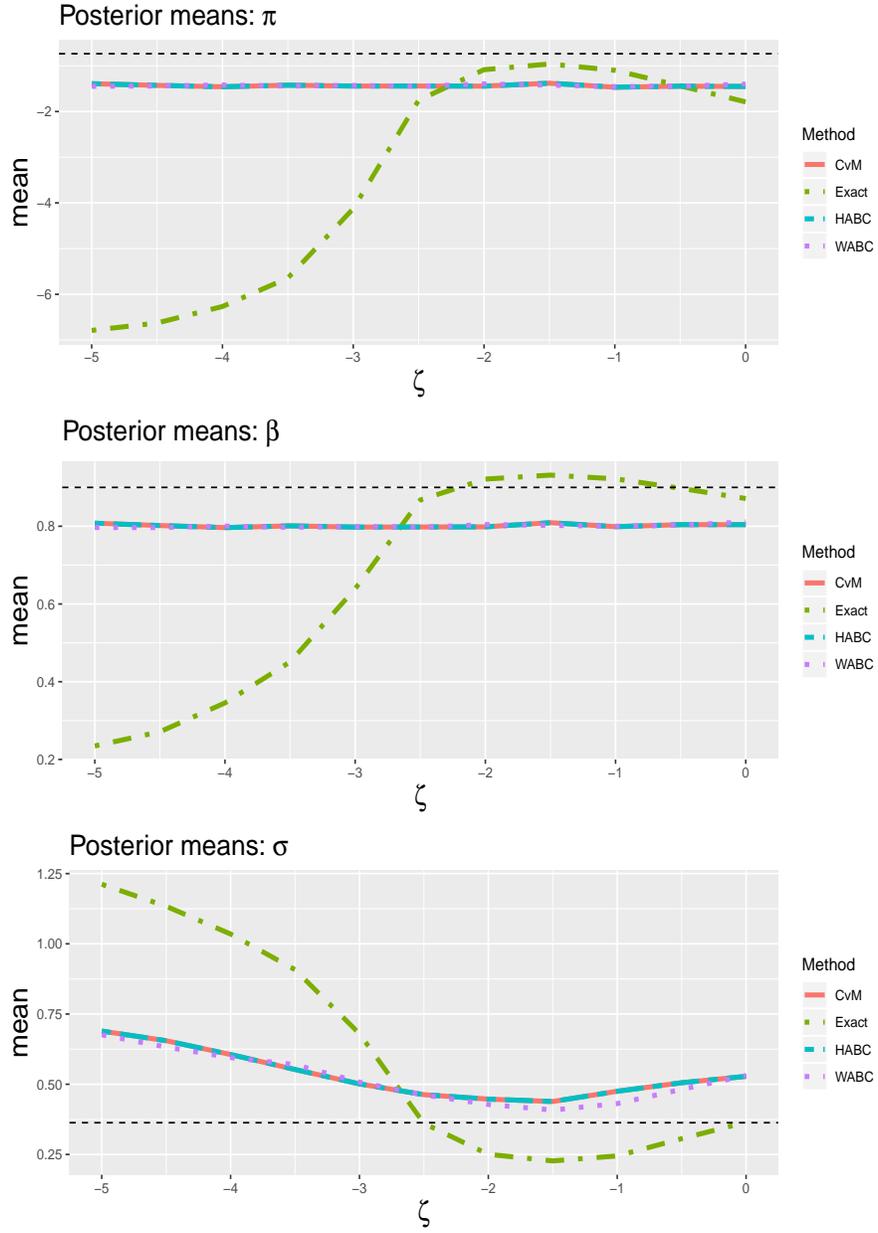}
	\caption{\small Posterior marginal means in the stochastic volatility example. The values used to generate the data are $\theta_\star=(-0.736,0.900,0.363)'$ with $\alpha=0.05$ and across the grid $\zeta\in[-5,0]$. The horizontal dashed lines represent the true values for the simulation. }
	\label{fig:svfig1}
\end{figure}

\begin{figure}[H]
	\centering
	\includegraphics[height=.75\textheight,width=.75\textwidth]{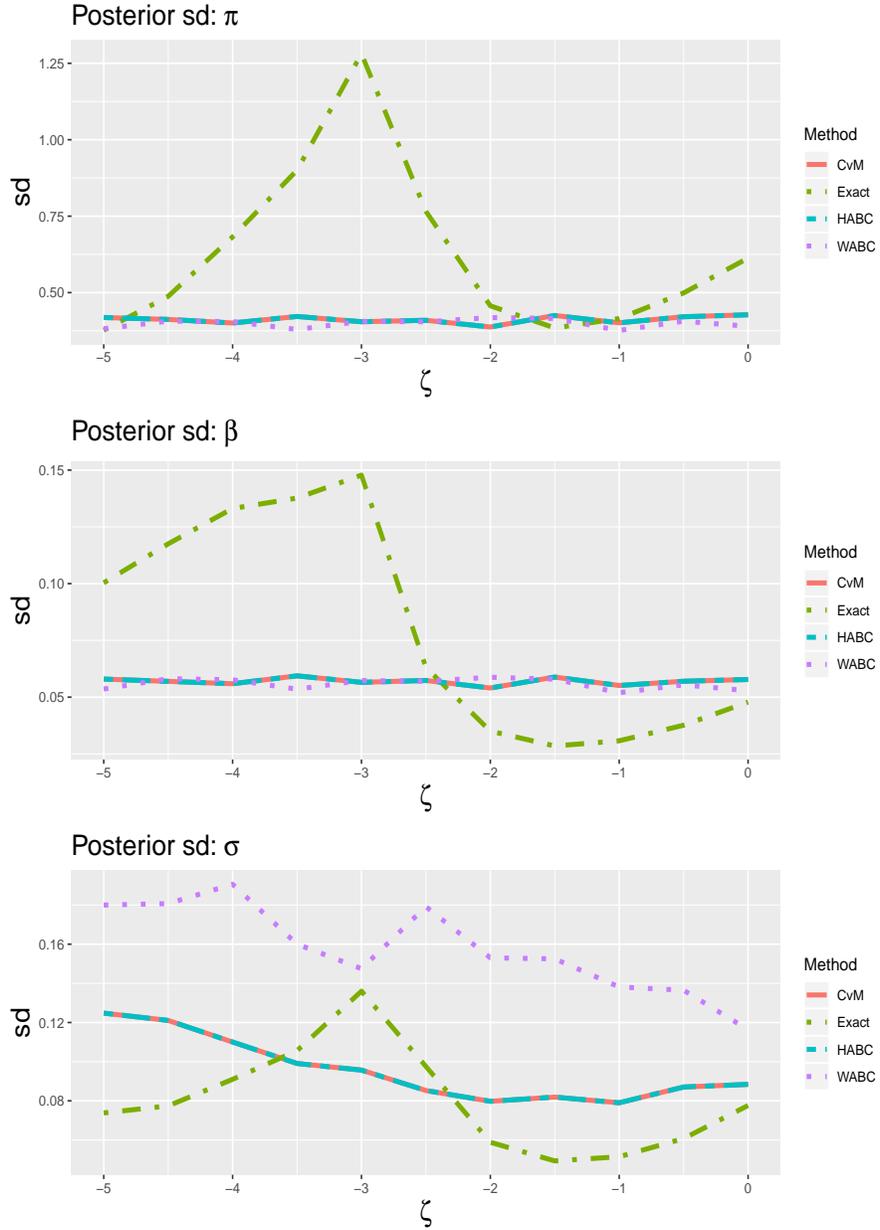}
	\caption{\small Posterior marginal standard deviations in the stochastic volatility example. The values used to generate the data are $\theta_\star=(-0.736,0.900,0.363)'$ with $\alpha=0.05$ and across the grid $\zeta\in[-5,0]$.}
	\label{fig:svfig2}
\end{figure}


\section{Discussion}
We have proposed a new approach to Bayesian inference, MD-ABC, based on comparing differences between empirical probability measures calculated from observed and simulated data. This method is particularly useful in models for which the calculation of likelihood functions is intractable. For particular choices of the norm, such as when the norm is chosen to be the Hellinger distance, MD-ABC point estimators are asymptotically efficient and, thus, are equivalent to those obtained from exact Bayesian methods (under the assumption that the model is correctly specified).

While originally motivated by situations where the likelihood is intractable, we have demonstrated that MD-ABC methods are a useful alternative to exact Bayesian methods if one is concerned with deviations from the underlying modeling assumptions. Moreover, we have shown that if the model is correctly specified the differences between MD-ABC and exact inference will generally be minor (depending on the choice of distance used in MD-ABC). As such, MD-ABC represents an approach to inference that is as efficient, or nearly so, as exact methods, but is simultaneously robust to certain types of deviations from the underlying modeling assumptions.

\vspace{0.5cm}

\noindent\textbf{Acknowledgments:} David T. Frazier gratefully acknowledges support by the Australian Research Council through grant DE200101070. We would also like to thank Professor Pierre Jacob for making the code needed to implement Wasserstein ABC publicly available. 

\bibliography{refs_mispecCDF}
\bibliographystyle{apalike}

\begin{appendix}
\section{Additional Examples}\label{app:examples}
In this section, we give a comparison of WABC and MD-ABC across two commonly employed examples in the ABC literature. The results demonstrate that, generally, MD-ABC can deliver inferences that are more efficient than those obtained by WABC. 

\subsection{Example One: g-and-k model  \label{sec:gk}}
In this section, we compare the performance of MD-ABC in the $g$-and-$k$ model, which is a standard test case in ABC-based inference (see, e.g., \citealp{drovandi2011likelihood}, \citealp{FP2012}, and \citealp{Bernton2017}). The $g$-and-$k$ distribution depends on three parameters, has no closed density/distribution, and is generally specified through its quantile function 
$$q \in(0,1) \mapsto a+b\left(1+0.8 \frac{1-\exp (-g z(q)}{1+\exp (-g z(q)}\right)\left(1+z(q)^{2}\right)^{k} z(q),$$where $z(q)$ refers to the $q$-th quantile of the standard normal distribution. The four parameters of the $g$-and-$k$ distributions have specific interpretations. The parameter $a$ represents the location, $b$ the scale, while $g$ and $k$ control the skewness and kurtosis, respectively. Following \citet{drovandi2011likelihood}, we consider the following priors on the parameters
\begin{flalign}\label{eq:gandk}
a\sim \mathcal{U}[0,10],\;b\sim \mathcal{U}[0,10],\;\;g\sim \mathcal{U}[0,10],\;k\sim \mathcal{U}[0,10],
\end{flalign}where $\mathcal{U}[0,1]$ denotes the uniform distribution on $[0,1]$.

The most common approach to ABC-based inference in this model is to choose as the summary statistics, the sample quantiles of the observed data. While a finite collection of quantiles for the data are reasonable summary statistics, \cite{Bernton2017} demonstrates that such choices yield poor inferences relative to WABC-based inference.

In this section, we compare the behavior of WABC and MD-ABC using the Hellinger (H-ABC) and CvM (CvM-ABC) distances. We follow the simulation design in \cite{Bernton2017}. We generate $n=100$ observations from the $g$-and-$k$ model in equation \eqref{eq:gandk} using parameters $a = 3, b = 1, g = 2, k = 0.5$. For all ABC approaches, we use the SMC sampler of \cite{del2012adaptive} with $N=1,024$ particles, and a total of $2\times 10^6$ simulations from the model. For both H-ABC and CvM-ABC,  we consider a simulated sample of length $m=200$, i.e., we take $m=2n$.

Figure \ref{fig:gandk:abc} shows the marginal posterior distributions obtained by the various approaches. The plots show that while their is no single dominant method, CvM-ABC seemingly performs best overall, with the method consistently being centered over the true values ($\theta=(3,1,2,.5)'$) and having less variable posteriors than the other methods. Generally, there is substantial agreement between the  H-ABC and WABC posteriors. Overall, these results demonstrate that these alternative distances are competitive with WABC, and can outperform WABC for certain parameters.

\begin{figure}[H]
	\centering
	\includegraphics[height=.55\textheight,width=1.1\textwidth]{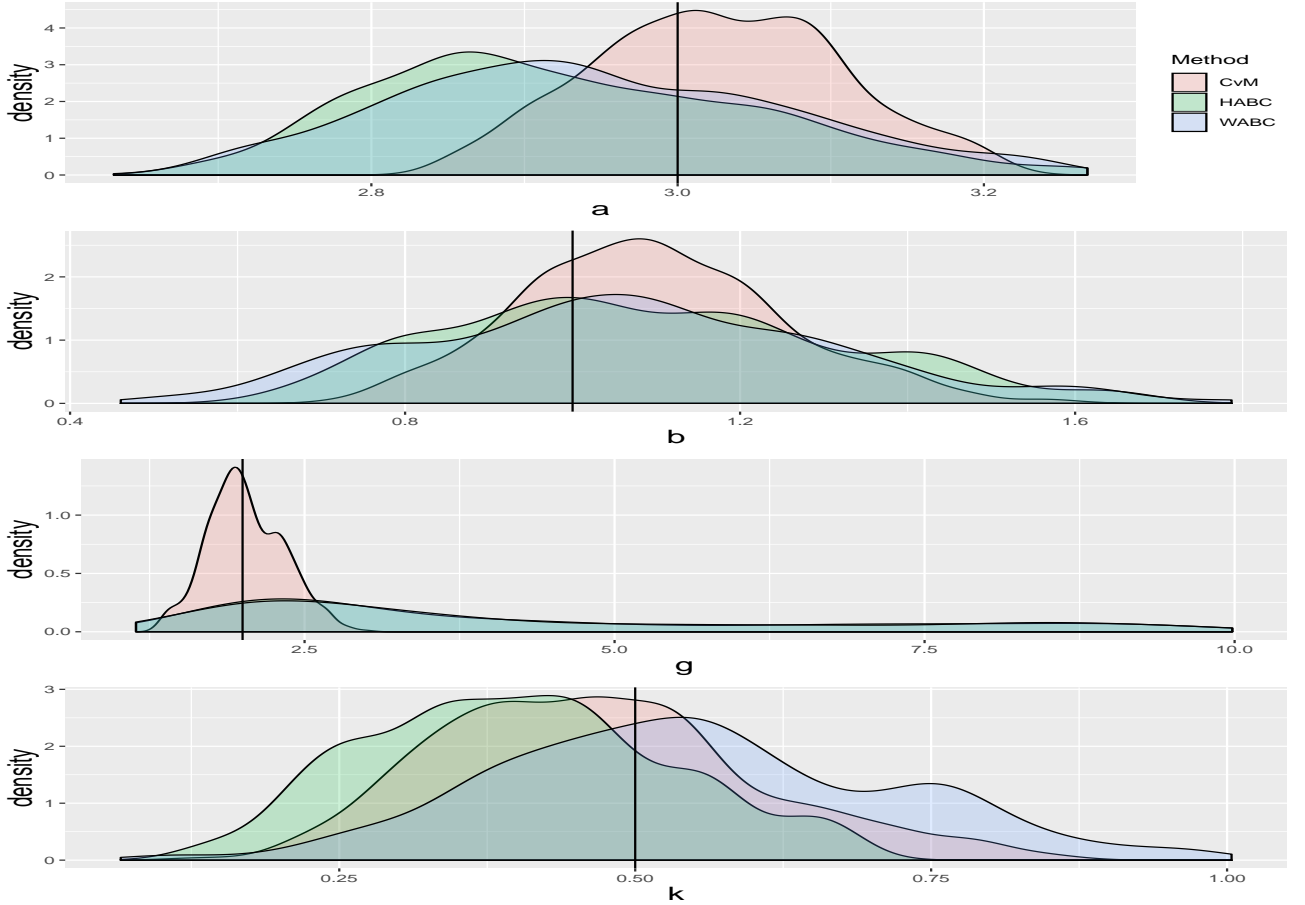}
	\caption{\small Posterior marginals in the $g$-and-$k$ model, with parameter values $\theta_\star=(a,b,g,k)'=(3,1,2,.5)'$. Each ABC method has a budget of $2*10^6$ model simulations. The vertical line indicates the true parameter value.  }
	\label{fig:gandk:abc}
\end{figure}

\subsection{Example Two: Queuing model \label{sec:queue}}

An additional test case that has appeared in the ABC literature is the M/G/1 queuing model (see e.g. \citealp{FP2012}). In the M/G/1 model,
customers arrive at a server with independent inter-arrival times $w_i$, which are assumed to be 
exponentially distributed with rate parameter $\theta_3$. The corresponding service times are 
independent random variables, $u_i$, which are uniformly distributed on
$[\theta_1,\theta_2]$. We observe only the inter-departure times $y_i$, given by
the process
$y_i = u_i + \max\{ 0, \sum_{j=1}^i w_j - \sum_{j=1}^{i-1} y_j\}$. The prior on
$(\theta_1, \theta_2-\theta_1, \theta_3)$ is Uniform on
$[0,10]^2\times[0,1/3]$.

We consider data generated according to the M/G/1 queuing model with $(\theta_1,\theta_2-\theta_1,\theta_3)' = (4,3,0.15)'$ and
$n=50$. We again compare H-ABC, CvM-ABC and WABC. All posteriors are approximated again using the same algorithm in Section \ref{sec:gk}, with a budget of $2\cdot10^6$ model simulations, and with $N=1,024$ parameters again retained to build the posterior approximation. For H-ABC and CvM-ABC, we consider a simulated data size of $m=150$, i.e., $m=3n$. 

The resulting posteriors are compared in Figure \ref{fig:queue:intermediate}. The results again demonstrate that by using alternative norms, we can obtain more efficient inferences than by using the WABC. In this context, across all three parameters, the H-ABC and CvM-ABC posteriors are quite similar, while the WABC posteriors have fatter tails. Generally, speaking all three approaches are well-centred over the true values used to generate the data $\theta_\star=(4,3,0.15)'$. Similar to the $g$-and-$k$ example, the CvM-ABC posterior performs particularly well across all parameters. 

\begin{figure}[H]
	\centering
	\includegraphics[height=.55\textheight,width=1.1\textwidth]{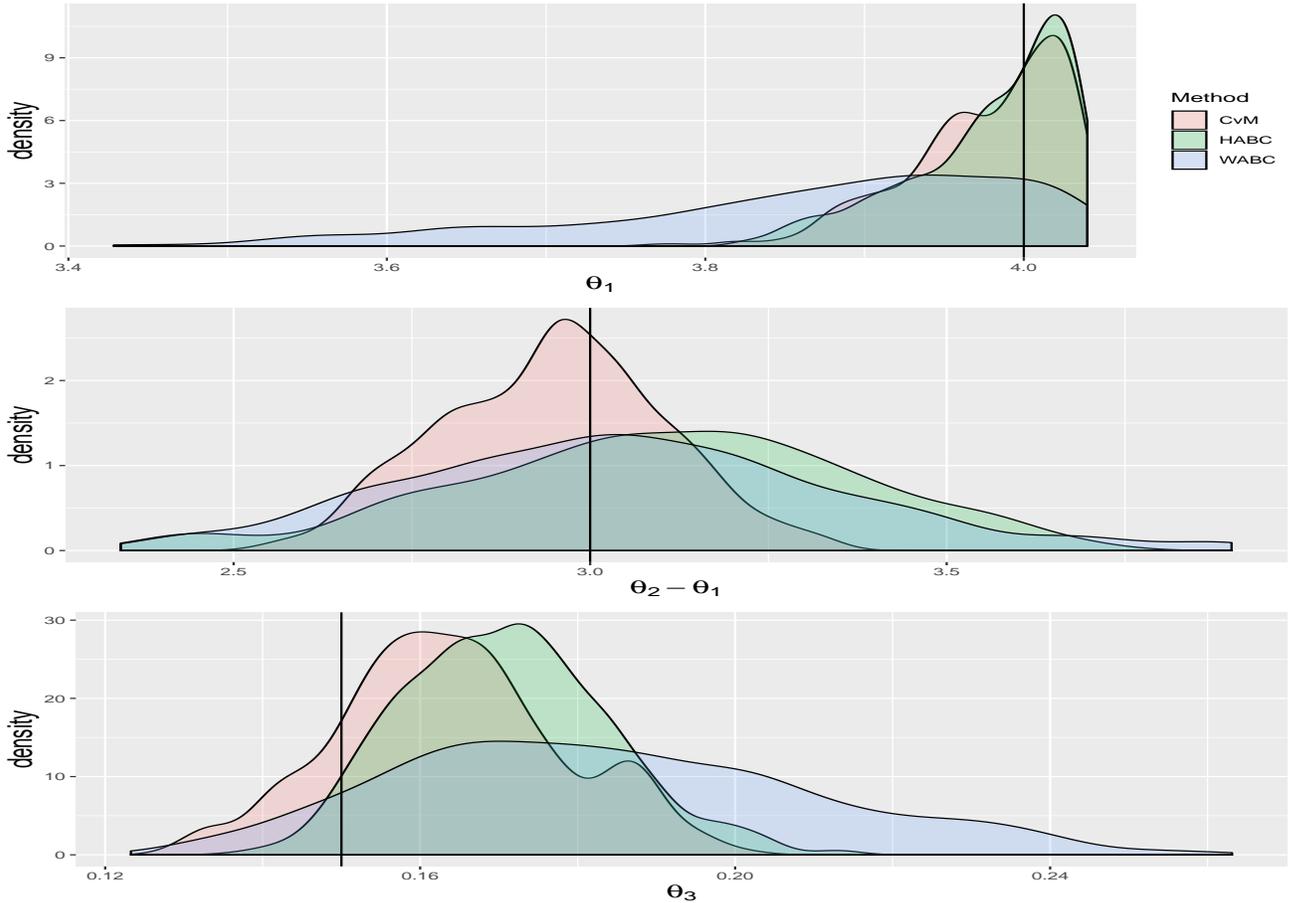}
	\caption{\small Posterior marginals in the M/G/1 queuing model of Section \ref{sec:queue}, with parameter values $\theta_\star=(4,3,0.15)'$. Each ABC method accounts for the constraint that $\theta_1$ has to be less than $\min_{i\in 1:n} y_i$, and has a budget of $2*10^6$ model simulations.}
	\label{fig:queue:intermediate}
\end{figure}
\section{Proofs of main results}
\begin{proof}[Proof of Theorem \ref{normal_thm}]
 Throughout the proof $C$ denotes a generic constant which may vary line to line. To simplify notation, throughout the remainder we drop the $\mathcal{P}$ notation from the norm $\|\cdot\|_\mathcal{P}$. First, define the following empirical processes:
\begin{flalign*}
\mathbb{G}_{n}^{0}=\sqrt{n}(\hat\mu_n-\mu_{\star}),\text{ and }\mathbb{G}_{n}=\sqrt{n}\left(\hat{\mu}_{\theta,n}-\mu_\theta\right).
\end{flalign*}Consider the decomposition:
\begin{equation*}
\begin{split}
\sqrt{n}(\hat{\mu}_{\theta,n}-\hat\mu_n)& =\sqrt{n}(\hat{\mu}_{\theta,n}-\mu_\theta)+\sqrt{n}(\mu_\theta-\mu_{\theta_\star})-(\hat\mu_n-\mu_{\theta_\star})
\\
& =\mathbb{G}_{n}+\sqrt{n} (\mu_\theta-\mu_{\theta_\star})-\mathbb{G}_{n}^{0}.
\end{split}%
\end{equation*}
For all $\theta\in\Theta$, apply the above decomposition to obtain
\begin{equation*}
\begin{split}
\mu^{(n)}_\theta  \left( \|\hat\mu_{\theta,n} - \hat\mu_n\|^2\leq \epsilon_n^2  \right)  &= \mu^{(n)}_\theta \left( \| \mathbb{G}_n - \mathbb{G}_n^0 +\sqrt{n}(\mu_\theta - \mu_\star)\|^2 \leq n \epsilon_n^2  \right) \\
& = \mu^{(n)}_\theta\left( \| \mathbb{G}_n\|^2 +2\langle \mathbb{G}_n , -\mathbb{G}_n^0 +\sqrt{n}(\mu_\theta - \mu_\star) \rangle \leq [n\cdot\epsilon_n^2 - \|\sqrt{n}\left(\mu_\theta-\mu_\star\right) -\mathbb{G}^0_n\|^2 ]\right) \\
&= \mu^{(n)}_\theta\left(\|\mathbb{G}_n\|^2+2\langle\mathbb{G}_n ,\sqrt{n}(\mu_\theta-\mu_\star)\rangle   \leq {[n\cdot\epsilon_n^2 -\|\sqrt{n}\left(\mu_\theta-\mu_\star\right) -\mathbb{G}^0_n\|^2 }{}  +2\langle\mathbb{G}_n, \mathbb{G}_n^0\rangle \right)
\end{split}
\end{equation*}
Now, on $\Omega_n =  \{ \| \mathbb{G}_n^0 \| \leq M_n/2\}$, with $M_n\rightarrow\infty$ and $M_n=o(n^{1/2})$, for $\hat\theta_n$ as in the statement of the theorem,
\begin{flalign}
 \|\sqrt{n}\left(\mu_\theta-\hat\mu_n\right)\|^2&=\|\sqrt{n}\left(\mu_\theta-\mu_\star\right) -\mathbb{G}^0_n\|^2 \nonumber\\&= n M_n^2 +  {n(\theta - \hat\theta_n)^{\intercal} H_\star(\theta - \hat\theta_n)}{ }    + O({n} \|\theta - \hat\theta_n\|^3+ \sqrt{n}\|\theta - \hat\theta_n\| \sqrt{n}M_n), \label{eq:expand1}
\end{flalign}
for a positive-definite matrix $H_\star=\mathbb{E}[ \dot\xi\dot\xi']$, where we note that the first derivative is zero at $\hat\theta_n$.
Likewise, Assumption \ref{ass:6} implies, 
\begin{flalign}
\langle\mathbb{G}_n ,\sqrt{n}(\mu_\theta-\mu_\star)\rangle=\langle\mathbb{G}_n,\sqrt{n}(\theta-\hat\theta_n)^{\intercal}\dot\xi\rangle+o\left(\|\mathbb{G}_n\|\|\theta-\hat\theta_n\|\right)\label{eq:expand2}
\end{flalign}
For $\gamma_n=o(1)$ and $\gamma_n>0$, with $\gamma_n$ the posterior concentration rate deduced by Theorem \ref{thm:one}, we restrict our attention to $\{\theta:\|\theta-\hat\theta_n\|\leq\gamma_n\}$. On the event $\Omega_n = \{ \|\mathbb{G}^0_n \| \leq M_n/2 \}$, where $M_n^2  = o({n})$,  apply the expansions in \eqref{eq:expand1}, \eqref{eq:expand2} and re-arrange terms to obtain,
	\begin{flalign*}
&\mu^{(n)}_\theta \left( \|\hat\mu_{\theta,n}-\hat\mu_n\|^2\leq \epsilon_n^2  \right) 
\\&  \leq  \mu^{(n)}_\theta\left( \| \mathbb{G}_n\|^2 +2\langle\mathbb{G}_n,\sqrt{n}(\theta-\hat\theta_n)^{\intercal}\dot\xi\rangle\leq {n\cdot\epsilon_n^2-(1+C\gamma_n)n(\theta - \hat\theta_n)^{\intercal} H_\star(\theta- \hat\theta_n) }  + C\gamma_n + \sqrt{n}\| \theta - \hat\theta_n\| \sqrt{n}M_n \right) \nonumber\\
&  + \mu^{(n)}_\theta\left( \|  \mathbb{G}_n\| > \gamma_n \sqrt{n}/2  \right) 
\end{flalign*}and 
\begin{flalign*}
&\mu^{(n)}_\theta \left( \|\hat\mu_{\theta,n}-\hat\mu_n\|^2\leq \epsilon_n^2  \right)\\&\geq
\mu^{(n)}_\theta\left(\| \mathbb{G}_n\|^2 +2\langle\mathbb{G}_n,\sqrt{n}(\theta-\hat\theta_n)^{\intercal}\dot\xi\rangle\leq {n\cdot\epsilon_n^2-(1-C\gamma_n)n(\theta - \hat\theta_n)^{\intercal} H_\star(\theta- \hat\theta_n) } - C\gamma_n - \sqrt{n}\| \theta - \hat\theta_n\| \sqrt{n}M_n \right)\nonumber\\
&  - \mu^{(n)}_\theta\left( \|  \mathbb{G}_n\| > \gamma_n \sqrt{n}/2  \right) 
\end{flalign*}

On the event $\Omega _{n}$, which by Assumption \ref{ass:eight}, has probability smaller than $\gamma_n $ by choosing $M$ large
enough, by Assumption \ref{ass:2}, for all $\delta>0$, and $\|\theta-\theta_\star\|\leq \delta$, $\mu_\theta^{(n)}\left[\Vert \mathbb{G}_{n}\Vert >u\right]\lesssim u^\kappa$ for all $0<u\leq \delta \sqrt{n}$.  Therefore, for some positive constant $C>0$ and for any $\epsilon>0$, we can choose $M$ large enough so that, 
\begin{flalign}
\mu^{(n)}_\theta\left( \Vert \mathbb{G}_{n}\Vert \geq M/2\right) &\leq CM^{-\kappa}\leq \epsilon \label{eq:cinq}
\end{flalign}
for all $\Vert {\theta }-{\theta }_{\star}\Vert \leq \gamma_{n}$.
Use \eqref{eq:cinq} to conclude that 
	\begin{flalign}\label{deve:likeli}
\mu^{(n)}_\theta &\left( \|\hat\mu_{\theta,n}-\hat\mu_n\|^2\leq \epsilon_n^2  \right) \pm C\gamma_n\\  \leq&  \mu^{(n)}_\theta\left( \| \mathbb{G}_n\|^2 +2\langle \mathbb{G}_n,\sqrt{n}(\theta-\hat\theta_n)^{\intercal}\dot\xi\rangle\leq {n\cdot\epsilon_n^2-(1+C\gamma_n)n(\theta - \hat\theta_n)^{\intercal} H_\star(\theta- \hat\theta_n)}  + C\gamma_n + \sqrt{n}\| \theta - \hat\theta_n\| \sqrt{n}M_n \right)\nonumber\\\geq
& \mu^{(n)}_\theta\left(\| \mathbb{G}_n\|^2 +2\langle \mathbb{G}_n,\sqrt{n}(\theta-\hat\theta_n)^{\intercal}\dot\xi\rangle\leq {n\cdot\epsilon_n^2-(1-C\gamma_n)n(\theta - \hat\theta_n)^{\intercal} H_\star(\theta- \hat\theta_n) } - C\gamma_n -\sqrt{n} \| \theta - \hat\theta_n\| \sqrt{n}M_n \right)
\end{flalign}
	Define $\theta\mapsto t:=\sqrt{n}(\theta-\hat\theta_n)$ and $h_n:=\epsilon_n\cdot \sqrt{n}$,
	\begin{flalign*}
	{[h^2_n -  (1-C\gamma)t^{\intercal} H_\star t ]}{}    + { n\cdot M_n^2} - \|t\| \sqrt{n}M_n- C\gamma_n  &\geq h^2_n
	-{ t^{\intercal} H_\star t}{} - C\|t\|\gamma[1- o(\|t\|)]	\end{flalign*}
	and 
	\begin{flalign*}
	{[h^2_n -  (1+C\gamma)t^{\intercal} H_\star t/2  ]}    + { n\cdot M_n^2} + \|t\| \sqrt{n}M_n + C\gamma_n\leq h^2_n
	-{ t^{\intercal} H_\star t}{} + C\|t\|\gamma[1+ o(\|t\|)]
	\end{flalign*}
	For $n$ large enough,
	\begin{equation*}
	\begin{split}
	\mu^{(n)}_\theta \left(\|\hat\mu_{\theta,n} - \hat\mu_n\|^2\leq \epsilon_n^2  \right)  &\leq \mu^{(n)}_\theta\left(\| \mathbb{G}_n+\langle t,\dot\xi\rangle\|^2 \leq h^2_n +C\|t\|\gamma[1+o(\|t\|)] \right)  +\gamma
	\\
	  &\geq \mu^{(n)}_\theta\left(\|\mathbb{G}_n+\langle t,\dot\xi\rangle \|^2\leq h^2_n-C\|t\|\gamma[1-o(\|t\|)]\right)-\gamma
	\end{split}
	\end{equation*}
	
Now, consider the change of variables $\theta\mapsto t=\sqrt{n}(\theta-\hat\theta_n)$
and let us analyze, for $R(t)$ denoting either $R_1(t)\asymp \|t\|\gamma[1+o(\|t\|)]$ and $R_2(t)\asymp-\|t\|\gamma[1-o(\|t\|)]$,
\begin{flalign*}
\int \mu^{(n)}_\theta\left(\|\hat\mu_{\theta,n} - \hat\mu_n\|^2\leq \epsilon_n^2  \right)\pi(\theta)\d\theta&\leq n^{-k_\theta/2}\int \mu^{(n)}_{t}\left(\| \mathbb{G}_n+\langle t,\dot\xi\rangle\|^2 \leq h_n +R_1(t) \right)\pi(\hat\theta_n+t/\sqrt{n})\d t\\&\geq n^{-k_\theta/2}
\int \mu^{(n)}_{t}\left(\| \mathbb{G}_n+\langle t,\dot\xi\rangle\|^2 \leq h_n +R_2(t) \right)\pi(\hat\theta_n+t/\sqrt{n})\d   t
\end{flalign*}
 Recall that $h_n=\epsilon_{n}^2{n}$ and define $P_{t}^{\ast }( \| \mathbb{G}_n+\langle t,\dot\xi\rangle\| \leq h_{n})=\mu^{(n)}_{t}( \| \mathbb{G}_n+\langle t,\dot\xi\rangle\| \leq
h_{n})/h^{k_\theta}_{n}$. For $h_{n}=o(1)$, by Assumptions \ref{ass:eight} and \ref{ass:nine}, 
\begin{flalign}
\frac{D_{n}}{h^{k_\theta}_{n}}&:= \int { P_{t}^{*}\left[ \| \mathbb{G}_n+\langle t,\dot\xi\rangle\| \leq h_{n}+R(t)\right]\pi(\hat\theta_n+t/\sqrt{n})}\d t\nonumber\\
&=\pi(\hat\theta_n)\int  
\varphi(t;0,V_{\theta_\star}) \d t+\int \left\{ P_{t}^{*}\left[ \| \mathbb{G}_n+\langle t,\dot\xi\rangle \| \leq h_{n}+R(t)\right]-\varphi(t;0,V_{\theta_\star})\right\}\nonumber\\
&\qquad \times \pi(\hat\theta_n+t/\sqrt{n}) \d t +o_p(1).\label{eq:convs}
\end{flalign}where $R(t)$ denotes either $R_1(t)$ or $R_2(t)$. We now show that $$ \frac{D_{n}}{h^{k_\theta}_{n}}=1+o_p(1).$$

If $h_{n}\rightarrow 0$, then
the result follows if $$\int \left\{ P_{t}^{\ast }[  \| \mathbb{G}_n+\langle t,\dot\xi\rangle\|\leq 
h_{n}+R(t)]-\varphi(t;0,V_{\theta_\star})\right\} \pi(\hat\theta_n+t/\sqrt{n})\d t=o_{p}(1),$$ for which a
sufficient condition is 
\begin{equation}
\int   \left\vert  P_{t}^{\ast }\left[ \| \mathbb{G}_n+\langle t,\dot\xi\rangle\|\leq 
h_{n}+R(t)\right]-\varphi(t;0,V_{\theta_\star}) \right\vert \pi(\hat\theta_{n}+t/\sqrt{n})\d t=o_{p}(1),
\label{show}
\end{equation}
We split the above integrals into three regions and analyze each piece: (i) $\{ \|t\| \leq M\}$, (ii) $\{ \sqrt{n}\gamma > \| t\| > M \}$ and (iii) $\{\gamma\cdot\sqrt{n} \leq \| t\| \}$, where $\gamma $ is arbitrarily small. In what follows we let $x:=\langle t,\dot\xi\rangle$. Before moving on, we note that, by Assumption \ref{ass:6} (iii), the three regions for $\|t\|$ can be translated into regions for $\|x\|$ as follows: for finite constants $0<C_1\leq C_2$, (i) $\{ \|x\| \leq C_2M\}$; (ii) $\{ C_2\sqrt{n}\gamma > \| x\| > C_1M \}$; (ii) $\{\gamma\cdot\sqrt{n}C_1 \leq \| x\| \}$
\medskip

\noindent\textbf{Region (i)}

\medskip

\noindent Over $ \|t\| \leq M$, the following equivalences are satisfied: 
\begin{flalign*}
\sup_{t:  \| t \| \leq
	M}& \mid \pi(\hat\theta_{n}+t/\sqrt{n})-\pi(\theta_\star) \mid =o_{p}(1),\\\sup_{t:  \| t \| \leq
	M}&|R(t)|=o_p(1).
\\\sup_{t:  \| x(t) \| \leq
	C_2M}& \mid P^{*}_{t}[ \| \mathbb{G}_n+x(t) \| \leq h_{n}+R(t)]-\varphi(t;0,V_{\theta_\star}) \mid =o_{p}(1), \end{flalign*}The first equation is satisfied by Assumption \ref{ass:nine} and the consistency of $\hat\theta_n$, which is proven in Lemma \ref{lemma:point_est}. The second term follows
from the definition of $R(t)$ and the fact that $\gamma_n=o(1)$. The second term follows from Assumption \ref{ass:eight} and the dominated convergence theorem. We can now conclude that equation \eqref{show}
is $o_{p}(1)$ over $ \| t \| \leq M$.

\medskip

\noindent\textbf{Region (ii)} 

\medskip

\noindent Over $M\leq \|t\| \leq \gamma \sqrt{n}$ the integral of the
second term can be made arbitrarily small for $M$
large enough, and it suffices to show that 
\begin{equation*}
\int_{M\leq \|t\| \leq \gamma \sqrt{n}} P_{t}^{\ast }[\| \mathbb{G}_{n}+x(t)\| \leq
h_{n}+R(t)]\pi(\hat\theta_n+t/\sqrt{n})\d t
\end{equation*}%
is finite.

When $\|x(t)\| >M$, $\| \mathbb{G}_n+x(t)\| \leq h_{n}+R(t)$ implies that $\| \mathbb{G}_n\| >\|x(t)\| /2\geq C\|t\|/2$ since $%
h_{n}=o(1)$ and for $\gamma^2=O(1/\sqrt{n})$,  $R(t)=o(1)$. Hence, using Assumption \ref{ass:2}, 
\begin{equation*}
P_{t}^{\ast }[\| \mathbb{G}_n+x(t)\| \leq h_{n}+R(t)]\leq  P_{t}^{\ast
}(\| \mathbb{G}_n\| >C\|t\| /2)\leq C{\|t\| ^{-\kappa }},
\end{equation*}%
which in turns implies that, for $\kappa>2$,  
\begin{equation*}
\int_{M\leq \|t\| \leq \gamma \sqrt{n}}P_t^{\ast }[\| \mathbb{G}_n+x(t)\| \leq
h_{n}+R(t)]\pi(\hat\theta_n+t/\sqrt{n}) \d t\leq C\int_{M\leq \|t\| \leq \gamma
	\sqrt{n}}{\|t\| ^{-\kappa}}\d t\leq M^{-\kappa +1}.
\end{equation*} 
\medskip

\noindent\textbf{Region (iii)} 

\medskip
 \noindent Over $\|t\| \geq \gamma \sqrt{n}$ the second term is again
negligible for $\gamma \sqrt{n}$ large. Our focus then becomes 
\begin{equation*}
\int_{ \|t\|\geq \gamma \sqrt{n}}P_{t}^{\ast
}[\| \mathbb{G}_n+x(t)\| \leq h_{n}+R(t)]\pi(\hat\theta_n+t/\sqrt{n})\d t.
\end{equation*}%
By Assumption \ref{ass:2}, for some $\kappa >2$ we bound this term as follows:
\begin{flalign*}
\frac{1}{h^{k_\theta}_{n}}&\int_{\|t\| \geq \gamma\sqrt{n}}\mu^{(n)}_{t}[\| x(t)+\mathbb{G}_n\| \leq
h_{n}+R(t)]\pi(\hat\theta_n+t/\sqrt{n})\d t\\&\leq \frac{1}{h^{k_\theta}_{n}}\int_{\|t\| \geq
	\gamma \sqrt{n}}\frac{  c(\hat\theta_n+t/\sqrt{n})
}{[1+\|t\| -h_{n}-R(t)]^{\kappa}}\pi(\hat\theta_n+t/\sqrt{n})\d t\\&\lesssim\frac{{n}^{k_\theta}}{h^{k_\theta}_{n}}\int_{\| \theta-\hat\theta_n\| \geq\gamma}\frac{{c(\theta)}
	\| \theta-\hat\theta_n\|  }{\{1+\sqrt{n}\| \theta-\hat\theta_n\| -h_{n}-R(\theta-\hat\theta_n)\}^{\kappa}}\pi(\theta)\d\theta
\end{flalign*}Since $\hat\theta_n=\theta_{\star}+O_{P}(1/\sqrt{n})$, by Lemma \ref{lemma:point_est} we have, for $n$
large,
\begin{flalign*}
\frac{{n}^{k_\theta}}{h^{k_\theta}_{n}}\int_{\| \theta-\theta_{\star}\| \geq\gamma/2}\frac{{c(\theta)} \pi(\theta)}{[1+\sqrt{n}\gamma\{1-O(\gamma)\}-h_{n}]^{\kappa}}\d\theta\lesssim  \frac{1}{\epsilon^{k_\theta}_{n}}\left\{\int {c(\theta)} \pi(\theta)\d\theta\right\}O(n^{k_\theta/2-\kappa})\lesssim O(n^{k_\theta/2-\kappa}\epsilon^{k_\theta}_{n})=o(1),
\end{flalign*}for $\kappa>k_\theta$, 
where Assumption \ref{ass:nine} ensure $\int {c(\theta)} \pi(\theta)\d\theta=O(1)$. Conclude that $D_n/h_n^{k_\theta}=1+o_p(1)$. 

All we have manipulated is the magnitude of $|R(t)|$, and so 
$$
D_n/h_n^{k_\theta}\asymp\frac{1}{h_n^{k_\theta}}\int \mu^{(n)}_t\left(\|\hat\mu_{\hat\theta_n+t/\sqrt{n},n} - \hat\mu_n\|^2\leq \epsilon_n^2  \right)\pi(\theta)\d\theta\asymp D_n/h_n^{k_\theta},
$$Conclude that 
$$
\frac{1}{h_n^{k_\theta}}\int \mu^{(n)}_t\left(\|\hat\mu_{\hat\theta_n+t/\sqrt{n},n} - \hat\mu_n\|^2\leq \epsilon_n^2  \right)\pi(\theta)\d\theta=1+o_p(1).
$$

To deduce the stated result, we use the above and control the behavior of
\begin{equation*}
\begin{split}
\E_{\Pi _{\epsilon }}\left[ \mathds{1}_{\sqrt{n}(\theta-\hat\theta_n)\in A}\right] & =\int \mathds{1}_{\{t\in A\}}\d\Pi _{\epsilon }\left( {%
	\theta } \mid y_{1:n}\right) \\
& =\int \mathds{1}_{\sqrt{n}(\theta-\hat\theta_n)\in A}1\!\mathrm{l}_{\Vert {\theta }-{\theta}_{\star}\Vert \leq \gamma _{n}}\d\Pi _{\epsilon }\left( {\theta } \mid {y_{1:n} }\right) +o_{P}(1) \\
& =\frac{\int_{\Vert {\theta }-{\theta }_{\star}\Vert \leq \gamma _{n}}\pi({%
		\theta })\mathds{1}_{\sqrt{n}(\theta-\hat\theta_n)\in A}\mu^{(n)}_\theta\left( \Vert \hat\mu_{\theta,n}-\hat\mu_n\Vert \leq \epsilon _{n}\right) \d{\theta }}{\int_{\Vert {%
			\theta }-{\theta }_{\star}\Vert \leq \gamma _{n}}\pi({\theta })\mu^{(n)}_\theta\left (\Vert \hat\mu_{\theta,n}-\hat\mu_n\Vert \leq \epsilon _{n}\right) \d{%
		\theta }}+o_{p}(1)\\&=\frac{N_n}{D_n}+o_p(1).
\end{split}%
\end{equation*}%
where the second equality uses the posterior concentration of $\Vert {\theta 
}-{\theta }_{0}\Vert $ at the rate $\gamma _{T}{\gg }
1/\sqrt{n}$. We have already demonstrated that $D_n/h_n^{k_\theta}=1+o_p(1)$. Similar arguments to the above yield 
\begin{flalign*}
\frac{N_n}{h_n^{k_\theta} }&=\int_{}\mathds{1}_{\sqrt{n}(\theta-\hat\theta_n)\in A}P^{\ast}_t\left[\|\mathbb{G}_n+x(t)\|\leq h_n\right]\d t+o(1)\\&=\int_{}\mathds{1}_{t\in A}P^{\ast}_{{t}}\left[\{\mathbb{G}_n+x(t)\}\in B_n\leq h_n\right]\d t+o(1)\\&=\int_{A}\varphi(t;0,V_{\theta_\star})\d t+o(1), 
\end{flalign*}where the last line follows from applying Assumption \ref{ass:nine} and the DCT. Together, these two equivalences yield the result. 
\end{proof}

\begin{proof}[Proof of Theorem \ref{mean_thm}] The proof follows along the same lines as that of Theorem \ref{normal_thm} and we only highlight the main differences. 

To deduce the stated result, consider the change of variables $t:=\sqrt{n}(\theta-\hat\theta_n)$ and obtain
\begin{equation*}
\begin{split}
\E_{\Pi _{\epsilon }}\left[ \theta\right] & =\int \theta\d\Pi _{\epsilon }\left({%
	\theta } \mid y_{1:n}\right) \\
& =\int \left(\hat\theta_n+t/\sqrt{n}\right)\d\Pi _{\epsilon }\left( {\hat\theta_n+t/\sqrt{n} } \mid {y_{1:n} }\right) +o_{p}(1)
\end{split}%
\end{equation*}%
which, following arguments in the proof of Theorem \ref{normal_thm}, yields
\begin{flalign*}
\sqrt{n}\left(\bar{\theta}-\hat\theta_n\right)=\frac{\int_{}t\pi(\hat\theta_n+t/\sqrt{n})P^{\ast}_{{t}}\left[\|\mathbb{G}_n+x(t)\|\leq h_n+R(t)\right] \d{t }}{\int_{}\pi({\hat\theta_n+t/\sqrt{n} })P^{\ast}_{{t}}\left[\|\mathbb{G}_n+x(t)\|\leq h_n+R(t)\right] \d{%
		t }}=\frac{N_n/h_n^{k_\theta}}{D_n/h_n^{k_\theta}}
\end{flalign*}

By the proof of Theorem \ref{normal_thm}, $D_n/h_n^{k_\theta}=1+o_p(1)$. Now, following the same arguments as in the Theorem \ref{normal_thm}, the term $N_n$ can be expressed as
$$
\frac{N_n}{h_n^{k_\theta}}=\int t[P_{t}^{\ast }[  \| \mathbb{G}_n+\langle t,\dot\xi\rangle\|\leq 
h_{n}+R(t)]-\varphi(t;0,V_{\theta_\star})]\pi(\hat\theta_n+t/\sqrt{n})\d t+\pi(\hat\theta_n)\int t \varphi(t;0,V_{\theta_\star})\d t+o_p(1).
$$
The stated result will follow if
\begin{equation}
\int   \|t\|\left\vert  P_{t}^{\ast }[  \| \mathbb{G}_n+\langle t,\dot\xi\rangle\|\leq 
h_{n}+R(t)]-\varphi(t;0,V_{\theta_\star})\right\vert \pi(\hat\theta_{n}+t/\sqrt{n})\d t=o_{p}(1),
\label{show1}
\end{equation}As in the proof of Theorem \ref{normal_thm}, the result can be proven by splitting the integral into three regions and showing that each piece is negligible. Given this, For the sake of brevity, we only provide a sketch that highlights the important differences.

\noindent For \textbf{Region (i)}, $\|t\|\leq M$,  we proceed exactly as in Theorem \ref{normal_thm}. 

\noindent For \textbf{Region (ii)},  Over $M\leq \|t\| \leq \gamma \sqrt{n}$, it suffices to show that 
\begin{equation*}
\int_{M\leq \|t\| \leq \gamma \sqrt{n}} \|t\|P_{t}^{\ast }[\| \mathbb{G}_{n}+x(t)\| \leq
h_{n}+R(t)]\pi(\hat\theta_n+t/\sqrt{n})\d t
\end{equation*}can be made finite for $M$ large.

When $\|t\|>M$, $\|x\|>CM$, for some $C>0$. Define $M':=CM$ and note that for $\|x(t)\| >M'$, $\| \mathbb{G}_n+x(t)\| \leq h_{n}+R(t)$ implies that $\| \mathbb{G}_n\| >\|x(t)\| /2$ since $%
h_{n}=o(1)$ and for $\gamma^2=O(1/\sqrt{n})$ $R(t)=o(1)$. Hence, using Assumption \ref{ass:6} and the inequality in \eqref{eq:cinq}, 
\begin{equation*}
\|t\|P_{t}^{\ast }\{\| \mathbb{G}_n+x(t)\| \leq h_{n}+R(t)\}\leq  \|t\|P_{t}^{\ast
}(\| \mathbb{G}_n\| >C\|t\| /2)\leq {C }/{\|t\| ^{\kappa-1 }},
\end{equation*}%
which in turns implies that 
\begin{equation*}
\int_{M\leq \|t\| \leq \gamma\sqrt{n}}P_t^{\ast }\{\| \mathbb{G}_n+x(t)\| \leq
h_{n}+R(t)\}\pi(\hat\theta_n+t/\sqrt{n}) \d t\leq C\int_{M\leq \|t\| \leq \gamma
	\sqrt{n}}\frac{1}{\|t\| ^{\kappa -1}}\d t\leq M^{-\kappa +2}.
\end{equation*} which can be made arbitrarily small for $M$ large enough and $\kappa>2$.

\noindent For \textbf{Region (iii)}, over $\|t\|\geq \gamma\sqrt{n}$, we proceed similar to Theorem \ref{normal_thm} but bound the remaining term as follows:
for some $\kappa >2$
\begin{flalign*}
\frac{1}{h_{n}^{k_\theta}}&\int_{\|t\| \geq \gamma \sqrt{n}}\|t\|\mu^{(n)}_{t}[\| x(t)+\mathbb{G}_n\| \leq
h_{n}+R(t)]\pi(\hat\theta_n+(t)/\sqrt{n})\d t\\&\leq \frac{1}{h^{k_\theta}_{n}}\int_{\|t\| \geq
	\gamma \sqrt{n}}\frac{  \|t\|c(\hat\theta_n+t/\sqrt{n})
}{[1+\|t\| -h_{n}-R(t)]^{\kappa}}\pi(\hat\theta_n+(t)/\sqrt{n})\d t\\&\lesssim\frac{n^{k_\theta}}{h^{k_\theta}_{n}}\int_{\| \theta-\hat\theta_n\| \geq\gamma}\frac{{c(\theta)}
	\| \theta-\hat\theta_n\|  }{\{1+\sqrt{n}\| \theta-\hat\theta_n\| -h_{n}-R(\theta-\hat\theta_n)\}^{\kappa}}\pi(\theta)\d\theta
\end{flalign*}The proof now follows exactly the same line as in the remainder of Theorem \ref{normal_thm}.

Putting the results together across the three regions demonstrates that $N_n/h_n^{k_\theta}=o_p(1)$.

\end{proof}

\begin{proof}[Proof of Theorem \ref{thm:bias}]

\noindent\textbf{Part (i).} Consider $r>0$ and $t\in\Theta$ arbitrary. Recall that the likelihood ratio for the simple perturbations satisfies
$$
\frac{\d\mu^{(n)}_{\theta_n,\zeta_n}}{\d\mu^{(n)}_{\theta_\star}}=1+\frac{t^{\intercal}\zeta_n}{\sqrt{n}},
$$ for $\zeta_n$ such that
$$
\lim_{n\rightarrow\infty}\int(\zeta_n-\zeta)^2\d\mu_{\star}=0,\text{ and }\zeta:=-H_\star^{-1}\dot\xi_{},
$$
and for $\dot\xi$ and $H_\star$ as in Assumption \ref{ass:6}. Under this construction, $\mu_{\theta_\star,0}=\mu_\star$, and since $t^{\intercal}\zeta_n/\sqrt{n}=O(n^{-1/2})$, the likelihood ratio ${\d\mu^{(n)}_{\theta_n,\zeta_n}}/{\d\mu^{(n)}_{\theta_\star}}$ is well-defined for $n$ large enough. By assumption, the functional $T_a$ is Fisher consistent for this class of simple perturbations. 

We consider the two neighborhoods separate.

\noindent\textbf{Hellinger neighborhood:} For arbitrary $r>0$ and $n$ large enough, choose a $\mu^{(n)}_{\theta_n,\zeta_n}\in\mathcal{B}_\has(\mu_\star,r/\sqrt{n})$ and consider the Hellinger distance between $\mu^{(n)}_{\theta_n,\zeta_n}$ and $\mu^{(n)}_{\theta_\star}$:
$$
\has^2\{\mu^{}_{\theta_n,\zeta_n},\mu_\star\}=\left\|\sqrt{\mu^{}_{\theta_n,\zeta_n}}-\sqrt{\mu^{}_{\star}} \right\|^2_2=\int\left(\sqrt{(1+t^{\intercal}\zeta_n/\sqrt{n})\d\mu_\star}-\sqrt{\d\mu_\star}\right)^2\d\lambda
$$Multiply by $n$ and expand around $t=0$, to obtain
\begin{flalign}\label{eq:hell_cons}
n\has^2\{\mu^{}_{\theta_n,\zeta_n},\mu_\star\}&=\frac{1}{8}t^{\intercal}\int\zeta_n\zeta_n^{\intercal}\d\mu_\star t+O(n\|t/\sqrt{n}\|^3)\\&\rightarrow\frac{1}{8}t^{\intercal}H_\star^{}t.
\end{flalign}
Using this, a lower bound of the maximum bias for $T_a$ can be obtained following the arguments in Theorem 5.3.5 of \cite{rieder2012robust}: 
\begin{flalign*}
\lim_{n\rightarrow\infty}\sup_{Q\in\mathcal{B}_\has\left(\mu_\star,r/\sqrt{n}\right)}&n\cdot|\tau\circ T_a(Q)-\tau(\theta_0)|^2\geq\\
\lim_{n\rightarrow\infty}\sup_{t\in\Theta:\mu_{\theta_n,\zeta_n}\in\mathcal{B}_\has\left(\mu_\star,r/\sqrt{n}\right)}&n\cdot|\tau\circ T_a(\mu_{\theta_n,\zeta_n})-\tau(\theta_0)|^2,
\end{flalign*}however, by Fisher consistency, and the differentiability of $\tau(\theta)$, $\sqrt{n}\left\{\tau\circ T_a(\mu_{\theta_n,\zeta_n})-\tau(\theta_\star)\right\}=\left[\d\tau(\theta_\star)/\d\theta\right]t+o(1)$, so that 
\begin{flalign*}
\lim_{n\rightarrow\infty}\sup_{t\in\Theta:\mu_{\theta_n,\zeta_n}\in\mathcal{B}_\has\left(\mu_\star,r/\sqrt{n}\right)}n\cdot|\tau\circ T_a(\mu_{\theta_n,\zeta_n})-\tau(\theta_0)|^2&\geq \lim_{n\rightarrow\infty}\sup_{t\in\Theta:\mu_{\theta_n,\zeta_n}\in\mathcal{B}_\has\left(\mu_\star,r/\sqrt{n}\right)}|\left[\d\tau(\theta_\star)/\d\theta\right]t|^2\\&\geq\lim_{n \rightarrow\infty}\sup_{t\in\Theta:t^{\intercal}H_\star^{}t\leq 8r^2}\left|\frac{\d}{\d\theta}\tau(\theta_0)t\right|^2\\&=8r^2\left[\frac{\d}{\d\theta}\tau(\theta_0)\right]H_\star^{-1}\left[\frac{\d}{\d\theta}\tau(\theta_0)\right]^{\intercal}
\end{flalign*}where the second inequality follows from the convergence in \eqref{eq:hell_cons}.\\

\noindent\textbf{CvM neighborhood:} For arbitrary $r>0$ and $n$ large enough, again choose a $\mu^{(n)}_{\theta_n,\zeta_n}\in\mathcal{B}_\lambda(\mu_\star,r/\sqrt{n})$. Calculate the CvM distance between $\mu^{(n)}_{\theta_n,\zeta_n}$ and $\mu^{(n)}_{\theta_\star}$: for $\lambda$ some $\sigma$-finite measure such that $$
\int \mu_\star(y)\left\{1-\mu_\star(y)\right\}\d\lambda(y)<\infty,
$$we have
\begin{flalign*}
n\mathcal{C}(\mu^{}_{\theta_n,\zeta_n},\mu_\star)&=\left\|{\mu^{}_{\theta_n,\zeta_n}}-{\mu^{}_{\star}} \right\|_{2,1}=n\int\left|\mu^{}_{\theta_n,\zeta_n}(y)-\mu_\star(y)\right|^2\lambda(\d y)\\
&=\int \left|\int\mathds{1}\left[x\leq y\right]t^{\intercal}\zeta_n(x)\mu_\star(\d x)\right|^2\lambda(\d y)\\&=
\int \left|\int\left\{\mathds{1}\left[x\leq y\right]-\mu_\star(\d y)\right\}t^{\intercal}\zeta_n(x)\mu_\star(\d x)\right|^2\lambda(\d y)
\end{flalign*}Define $\Delta_n(y):=\int \left\{\mathds{1}\left[x\leq y\right]-\mu_\star(\d y)\right\}\zeta_n(x)\mu_\star(\d x)$ and note that, $\Delta(y):= \int \left\{\mathds{1}\left[x\leq y\right]-\mu_\star(\d y)\right\}\zeta(x)\mu_\star(\d x)$. By the definition of the simple perturbations, 
\begin{flalign*}
\int \left|\Delta_n(y)-\Delta(y)\right|^2\lambda(\d y)&\leq \E\left|\Delta_n-\Delta\right|^2\int\mu_\star(y)\{1-\mu_\star(y)\}\d\lambda(y)=o(1).
\end{flalign*}Conclude that
\begin{flalign*}
n\mathcal{C}(\mu^{}_{\theta_n,\zeta_n},\mu_\star)=&\int |t^{\intercal}\Delta_n(y)|^2\d\lambda(y)\rightarrow t^{\intercal}{H}_\star t,
\end{flalign*}where 
$$
H_\star=\mathbb{E}[\Delta\Delta^{\intercal}]:=\int\int\left\{\mu_\star\left(z\wedge y \right)-\mu_\star(y)\mu_\star(z)\right\}\zeta(y)\zeta(z)^{\intercal}\lambda(\d y)\lambda(\d z).
$$
Using this, the result follows exactly the same arguments as for the Hellinger neighborhood, and is omitted. 

\noindent\textbf{Part (ii)} Again, we break the proof up into the different neighborhoods. \\

\noindent\textbf{Hellinger neighborhood:} Define the functional $T_h(Q):=\argmin_{\theta\in\Theta}\has\{\mu_\theta,Q\}$. By Theorem 6.3.4 of \cite{rieder2012robust}, part (a), for some $\delta_n=o(1)$, $T_h$ is Hellinger differentiable for all $\theta\in\{\theta:\|\theta-\theta_\star\|\leq \delta_n\}$ and has influence curve $\zeta_\theta:=H_{\theta}^{-1}\dot\xi_\theta$, where, in this case, $H_\theta$ is the Fisher information at $\theta$ and $\dot\xi_\theta$ the score (at $\theta$). Consequently, $T_h$ is Fisher consistent for any $t\in\bar{\Theta}$, with $\bar{\Theta}\subseteq\Theta$ and compact. Moreover, for $\zeta_\theta$ as above, and for all sequences $Q_n\in\mathcal{B}_\has(\mu_{\theta_\star},r/\sqrt{n})$, by Theorem 6.4.5 of \cite{rieder2012robust}, 
$$
\sqrt{n}\{T_h(Q_n)-\theta_\star\}=2\sqrt{n}\int \zeta_{\theta_\star}\d Q_n+o(1). 
$$From above expansion, and the bounded differentiability of $\tau(\theta)$, for $\d\tau_\star:= \frac{\d\tau(\theta_\star)}{\d\theta}$,
\begin{flalign*}
\frac{1}{2}\left\{\tau\circ T_h(Q_n)-\tau(\theta_\star)\right\}^2&=\left\{\d\tau_\star\int \zeta_{\theta_\star}\d Q_n^{1/2}\left(\d Q_n^{1/2}-\d \mu_{\theta_\star}^{1/2}\right)-\d\tau_\star\int \zeta_{\theta_\star}\d\mu_{\theta_{\star}}^{1/2}\left(\d Q_n^{1/2}-\d \mu_{\theta_\star}^{1/2}\right)\right\}^2+o(1)
\\&\leq\left|\d\tau_\star\int \zeta_{\theta_\star}\d Q_n^{1/2}\left(\d Q_n^{1/2}-\d \mu_{\theta_\star}^{1/2}\right)\right|^2 +\left|\d\tau_\star\int \zeta_{\theta_\star}\d\mu_{\theta_{\star}}^{1/2}\left(\d Q_n^{1/2}-\d \mu_{\theta_\star}^{1/2}\right)\right|^2\\&+2\left|\d\tau_\star\int \zeta_{\theta_\star}\d Q_n^{1/2}\left(\d Q_n^{1/2}-\d \mu_{\theta_\star}^{1/2}\right)\right|\left|\d\tau_\star\int \zeta_{\theta_\star}\d\mu_{\theta_{\star}}^{1/2}\left(\d Q_n^{1/2}-\d \mu_{\theta_\star}^{1/2}\right)\right|+o(1)\\&=L_1+L_2+2L_3+o(1).
\end{flalign*}

Consider the $L_1$ term and obtain
\begin{flalign*}
L_1\leq \left|\d\tau_\star\int \zeta_{\theta_\star}\zeta_{\theta_\star}^{\intercal}\d Q_n\d\tau_\star^{\intercal}\right|\left|\int \{\d Q_n^{1/2}-\d\mu_{\theta_\star}^{1/2}\}^2\right|^{}\leq \left|\d\tau_\star\int \zeta_{\theta_\star}\zeta_{\theta_\star}^{\intercal}\d Q_n\d\tau_\star^{\intercal}\right|\frac{r^2}{{n}}\leq B^*\frac{r^2}{{n}}+o(n^{-1}),
\end{flalign*}where the first inequality follows from Cauchy-Schwartz, the second inequality from the fact that $Q_n\in\mathcal{B}_\has\left(\mu_\star,r/\sqrt{n}\right)$, and the third inequality from the DCT. A similar argument yields equivalent bounds for $L_2$ and $L_3$ so that we obtain 
$$
\limsup_{n \rightarrow \infty}n\left\{\tau\circ T_h(Q_n)-\tau(\theta_\star)\right\}^2\leq 8 B^*\frac{r^2}{{n}}.
$$

To show the result for $\bar{T}(Q_n)$, recall that, by Theorem 6.4.5 of \cite{rieder2012robust}, 
$$
\sqrt{n}\{T_h(Q_n)-\theta_\star\}=2\sqrt{n}\int \zeta_{\theta_\star}\d Q_n+o(1).
$$ and by Theorem \ref{mean_thm} 
$$
\sqrt{n}\|\left\{\bar{T}(\hat{\mu}_n)-\theta_\star\right\}-\left\{T_h(\hat{\mu}_n)-\theta_\star\right\}\|=o_p(1).
$$Taken jointly, we have that  
$$
\sqrt{n}\left\{\bar{T}(Q_n)-\theta_\star\right\}-2\sqrt{n}\int \zeta_{\theta_\star}\d Q_n=o(1).
$$for all $Q_n\in\mathcal{B}_\has\left(\mu_\star,r/\sqrt{n}\right)$. Replacing the above derivations with this expansion does not alter the conclusions.

\noindent\textbf{CvM neighborhood:} The proof follows similarly to the case of the Hellinger neighborhood, and so we only sketch the main arguments.

Define the functional $T_c(Q):=\argmin_{\theta\in\Theta}\mathcal{C}\{\mu_\theta,Q\}$. By Theorem 6.3.4 of \cite{rieder2012robust}, part (b), for some $\delta_n=o(1)$, $T_c$ is CvM differentiable for all $\theta\in\{\theta:\|\theta-\theta_\star\|\leq \delta_n\}$ and has influence curve $\zeta_\theta:=H_{\theta}^{-1}\Delta_\theta$, where $\Delta_\theta(y):=\int\{\mathds{1}[x\leq y]-\mu_\theta(\d x)\}\dot\xi_\theta(x)\mu_\theta(\d x)$ and $H_{\theta}:=\E\left[\Delta_\theta\Delta_\theta^{\intercal}\right]$ (positive-definite for all $\theta\in\{\theta:\|\theta-\theta_\star\|\leq\delta_n\}$ by Assumption \ref{ass:6}). For $\zeta_\theta$ as above, and for all sequences $Q_n\in\mathcal{B}_\lambda(\mu_{\theta_\star},r/\sqrt{n})$, by Proposition 6.4.7 of \cite{rieder2012robust}, 
$$
\sqrt{n}\{T_c(Q_n)-\theta_\star\}=\sqrt{n}\int \left\{Q_n(y)-\mu_\star(y)\right\}\zeta_{\theta_\star}(y)\lambda(\d y)+o(1). 
$$From the above expansion, and the bounded differentiability of $\tau(\theta)$, for $\d\tau_\star:={\d\tau(\theta_\star)}/{\d\theta}$,
\begin{flalign*}
\frac{1}{2}\left|\tau\circ T_h(Q_n)-\tau(\theta_\star)\right|^2&=\left|\d\tau_\star\int \left\{Q_n(y)-\mu_\star(y)\right\}\zeta_{\theta_\star}(y)\lambda(\d y)\right|^2+o(1)
\\&\leq\left|\d\tau_\star\int \zeta_{\theta_\star}(y)\zeta_{\theta_\star}(y)^{\intercal}\lambda(\d y)\d\tau_\star^{\intercal}\right| \left|\int\{Q_n(y)-\mu_\star(y)\}^2\lambda(\d y)\right|\\&=\d\tau_\star H_\star^{-1}\d\tau_\star^{\intercal}\left|\int|Q_n(y)-\mu_\star(y)|^2\lambda(\d y)\right|+o(1)\\&=\frac{r^2}{n}\d\tau_\star H_\star^{-1}\d\tau_\star^{\intercal}.
\end{flalign*}The first inequality follows by Cauchy-Schwartz, the second from the definition of $\zeta_{\theta_\star}$, and the third from the fact that $Q_n\in\mathcal{B}_{\mathcal{C}}\left(\mu_\star,r/\sqrt{n}\right)$. The remainder of the proof follows the same as in the Hellinger case.

\end{proof}

\subsection{Additional Lemmas}
\begin{lemma}\label{lemma:point_est}Under Assumptions \ref{ass:1}, \ref{ass:2}, \ref{ass:5} and Assumptions \ref{ass:5a}-\ref{ass:seven}, for $\hat{\theta}_n$ satisfying $\|\mu_{\hat\theta_n}-\hat\mu_n\|_{\mathcal{P}}\leq \inf_{\theta\in\Theta}\|\mu_\theta-\hat\mu_n\|_{\mathcal{P}}+o_{p}(1/\sqrt{n})$,
	$$
	\sqrt{n}\|\hat\theta_n-\theta_\star\|=O_p(1).
	$$	
	
\end{lemma}
\begin{proof}[Proof of Lemma \ref{lemma:point_est}]
The result follows from standard arguments. From Assumption \ref{ass:1}, 
$$
\sup_{\theta\in\Theta}\left|\|\mu_\theta-\mu_n\|_{\mathcal{P}}-\|\mu_\theta-\mu_\star\|\right|=o_p(1).
$$From Assumption \ref{ass:5}(ii), 
$$
\|\mu_{\theta_\star}-\mu_\star\|_{\mathcal{P}}<\inf_{\theta\in\Theta}\|\mu_\theta-\mu_\star\|_{\mathcal{P}}
$$Consistency then follows from arguments miring those in Theorem 5.7 of \cite{van2000asymptotic}. 

All that remains is to establish $\sqrt{n}$-consistency. From Assumption \ref{ass:6}, 
\begin{flalign*}
o(\|\hat\theta_n-\theta_\star\|)&=
\|\mu_{\hat{\theta}_n}-\mu_\star-\langle (\hat\theta_n-\theta_\star),\dot\xi\rangle\|\geq \|\langle (\hat\theta_n-\theta_\star),\dot\xi\rangle\|-\|\hat\mu_n-\mu_{\hat\theta_n}\|-\|\hat\mu_n-\mu_\star\|
\end{flalign*}It then follows that 
\begin{flalign*}
\|\langle (\hat\theta_n-\theta_\star),\dot\xi\rangle\|[1-o_p(1)]\leq  \|\hat\mu_n-\mu_{\hat\theta_n}\|+\|\hat\mu_n-\mu_\star\|\leq o_p(n^{-1/2})+O_p(n^{-1/2})
\end{flalign*}
where the first inequality comes from re-arranging terms in the previous equation, and the second comes from applying $\|\mu_{\hat\theta_n}-\hat\mu_n\|_{\mathcal{P}}\leq \inf_{\theta\in\Theta}\|\mu_\theta-\hat\mu_n\|_{\mathcal{P}}+o_{p}(1/\sqrt{n})$, and Assumption \ref{ass:1}. Applying Assumption \ref{ass:6} (iii) then yields 
$$
\|\hat\theta_n-\theta_\star\|[1-o_p(1)]\leq o_p(n^{-1/2})+O_p(n^{-1/2}).
$$
\end{proof}
\end{appendix}

\end{document}